\documentclass[10pt,journal,compsoc]{IEEEtran}
\IEEEoverridecommandlockouts

\newif\iftr 
\trfalse

\if CLASSOPTIONcompsoc
\usepackage[nocompress]{cite}

\else
\usepackage{cite}
\fi

\usepackage{xcolor}
\def\BibTeX{{\rm B\kern-.05em{\sc i\kern-.025em b}\kern-.08em
		T\kern-.1667em\lower.7ex\hbox{E}\kern-.125emX}}

\usepackage[euler]{textgreek}
\usepackage{amssymb,amsmath,amsthm}
\usepackage[utf8x]{inputenc}
\usepackage{listings}
\usepackage{color}
\usepackage{epsfig}
\usepackage{url}
\usepackage{amsmath}
\usepackage{cite}
\usepackage{graphicx}
\usepackage{color}
\usepackage{amssymb,amsmath,dsfont}
\usepackage{amsthm}
\usepackage{algorithm}
\usepackage{algpseudocode}
\usepackage{epstopdf}
\usepackage{lipsum}
\usepackage{romannum}

\usepackage{pgfplots}
\pgfplotsset{compat=1.13}
\usepackage{tikz}
\usepackage{graphicx}
\usepackage{breqn}

\usepackage{graphics}

\newtheorem{lemma}{Lemma}
\newtheorem{corollary}{Corollary}

\newtheorem{theorem}{Theorem}

\newtheorem{sub-goal}{Sub-goal}

\usepackage{csquotes}
\usepackage{amsmath}
\usepackage{amssymb}
\usepackage{subcaption}
\usepackage{graphicx}
\usepackage{accents}
\usepackage{mathtools}
\usepackage[normalem]{ulem}
\usepackage{dirtytalk}
\usepackage{mathrsfs}
\usepackage{amsmath}

\graphicspath{figures/}

\usepackage{mathtools}

\usepackage{xspace}

\newcommand{\set}[1]{\left\{{#1}\right\}}

\newcommand{\lmax}{l^{\textit{max}}}
\newcommand{\isdef}{\ensuremath{\overset{\textit{def}}{=}}}

\usepackage[textwidth=1.8cm]{todonotes}
\newcommand{\sref}[1]{Section~\ref{#1}}
\newcommand{\thref}[1]{Theorem~\ref{#1}}

\renewcommand{\lmax}{l^{\scriptsize\textrm{max}}}
\newcommand{\tot}{\textrm{tot}}
\renewcommand{\isdef}{\ensuremath{\overset{\scriptsize\textit{def}}{=}}}
\newcommand{\lp}{\ensuremath{\left(}}
\newcommand{\rp}{\ensuremath{\right)}}
\newcommand{\lb}{\ensuremath{\left[}}
\newcommand{\rb}{\ensuremath{\right]}}

\newcommand{\mand}{\mbox{ and }}
\newcommand{\mwith}{\mbox{ with }}

\newcommand{\mymod}{\hspace{-2mm}\mod}

\newcommand{\tarr}{t_1^{\scriptsize\textrm{arr}}}
\newcommand{\tdep}{t_1^{\scriptsize\textrm{dep}}}
\newcommand{\dint}{D^{\scriptsize\textrm{TS}}}

\newcommand{\dsoni}{D^{{\scriptsize\textrm{Soni-et-al}}}}
\newcommand{\ctsoni}{C^{{\scriptsize\textrm{Soni-et-al}}}}
\newcommand{\dboyer}{D^{\scriptsize\textrm{Boyer-et-al}}}
\newcommand{\betaAnne}{\beta^{\scriptsize\textrm{Bouillard}}}

\newcommand{\rmin}{R^{{\scriptsize\textrm{min}}}}
\newcommand{\rmax}{R^{{\scriptsize\textrm{max}}}}

\newcommand{\tmin}{T^{{\scriptsize\textrm{min}}}}
\newcommand{\tmax}{T^{{\scriptsize\textrm{max}}}}

\newcommand{\lmod}{l^{{\scriptsize\textrm{mod}}}}
\newcommand{\lmodm}{l^{{\scriptsize\textrm{mod-$\epsilon$}}}}
\newcommand{\lmodp}{l^{{\scriptsize\textrm{mod+$\epsilon$}}}}

\newcommand{\oldservice}{\beta^{\textrm{\scriptsize old}}}

\newcommand{\newserviceBar}{\bar{\beta}^{\textrm{\scriptsize new}}}
\newcommand{\newservice}{\beta^{\textrm{\scriptsize new}}}
\newcommand{\newservicehat}{\hat{\beta}^{\textrm{\scriptsize new}}}

\newcommand{\mdelta}{d^{\scriptsize\textrm{max}}}

\newcommand{\phiMinL}{\phi_{i,j}^{\scriptsize\textrm{minLatency}}}
\newcommand{\phiMaxR}{\phi_{i,j}^{\scriptsize\textrm{maxRate}}}
\newcommand{\phicon}{\phi_{i,j}^{\scriptsize\textrm{concave}}}
\newcommand{\gammaMinL}{\gamma_i^{\scriptsize\textrm{minLatency}}}
\newcommand{\gammaMaxR}{\gamma_i^{\scriptsize\textrm{maxRate}}}
\newcommand{\gammacon}{\gamma_i^{\scriptsize\textrm{convex}}}

\newcommand{\betaMaxR}{\beta_i^{\scriptsize\textrm{maxRate}}}
\newcommand{\betaMinL}{\beta_i^{\scriptsize\textrm{minLatency}}}
\newcommand{\betaaff}[1]{\bar{\beta}^{\scriptsize\textrm{convex},#1}}
\newcommand{\betacon}[1]{\beta^{\scriptsize\textrm{concave},#1}}
\newcommand{\betaconvex}[1]{\beta^{\scriptsize\textrm{convex},#1}}

\newcommand{\deltaOld}{\delta^{\scriptsize\textrm{old}}}
\newcommand{\deltaMinL}[1]{\delta^{\scriptsize\textrm{minLatency},#1}}
\newcommand{\deltaMaxR}[1]{\delta^{\scriptsize\textrm{maxRate},#1}}

\begin{document}
	\pagenumbering{arabic}
	\pagestyle{plain}
	\title{Deficit Round-Robin: A Second Network Calculus Analysis}

%

	\author{\IEEEauthorblockN{Seyed Mohammadhossein Tabatabaee}
	\IEEEauthorblockA{\textit{EPFL}\\
		Lausanne, Switzerland \\
		hossein.tabatabaee@epfl.ch}
	\and
	\IEEEauthorblockN{Jean-Yves Le Boudec}
	\IEEEauthorblockA{\textit{EPFL}\\
		Lausanne, Switzerland \\
		jean-yves.leboudec@epfl.ch}
}

\author{Seyed~Mohammadhossein~Tabatabaee,~\IEEEmembership{Member,~IEEE,}
	and, Jean-Yves~Le~Boudec,~\IEEEmembership{Fellow,~IEEE}
\IEEEcompsocitemizethanks{\IEEEcompsocthanksitem EPFL in Lausane, Switzerland\protect\\
	E-mail: hossein.tabatabaee@epfl.ch, jean-yves.leboudec@epfl.ch}
}

\IEEEtitleabstractindextext{%
\begin{abstract}
	Deficit Round-Robin (DRR) is a widespread scheduling algorithm that provides fair queueing with variable-length packets. Bounds on worst-case delays for DRR were found by Boyer et al., who used a rigorous network calculus approach and characterized the service obtained by one flow of interest by means of a convex strict service curve. These bounds do not make any assumptions on the interfering traffic hence are pessimistic when the interfering traffic is constrained by some arrival curves. For such cases, two improvements were proposed. The former, by Soni et al., uses a correction term derived from a semi-rigorous heuristic; unfortunately, these bounds are incorrect, as we show by exhibiting a counter-example. The latter, by Bouillard, rigorously derive convex strict service curves for DRR that account for the arrival curve constraints of the interfering traffic. In this paper, we improve on these results in two ways. First, we derive a non-convex strict service curve for DRR that improves on Boyer et al. when there is no arrival constraint on the interfering traffic. Second, we provide an iterative method to improve any strict service curve (including Bouillard's) when there are arrival constraints for the interfering traffic. As of today, our results provide the best-known worst-case delay bounds for DRR. They are obtained by using the method of the pseudo-inverse.

\end{abstract} 
	
	}


	\maketitle
		
	\IEEEdisplaynontitleabstractindextext
	\IEEEpeerreviewmaketitle
	
	\setcounter{page}{1}
\IEEEraisesectionheading{\section{Introduction} \label{sec:intro}}
Deficit Round-Robin (DRR) \cite{DRR} is a scheduling algorithm that is often used for scheduling tasks, or packets, in real-time systems or communication networks. It is a variation of Weighted Round-Robin (WRR) that enables flows with variable packet lengths to fairly share the capacity. The capacity is shared among several clients or queues by giving each of them a quantum value and  by providing more service to those with larger quantum. DRR is widely used because it exhibits a low complexity, $O(1)$, provided that an allocated quantum is no smaller than the maximum packet size; and it can be implemented in very efficient ways, such as the Aliquem implementation \cite{Aliquem}.


We are interested in delay bounds for the worst case, as is typical in the context of deterministic networking. To this
end, a standard approach is network calculus. Specifically, with network calculus, the service offered to a flow of interest by a system is abstracted by means of a service curve. A bound on the worst-case delay is obtained by combining the service curve with an arrival curve for the flow of interest. An arrival curve is a constraint on the amount of data that the flow of interest can send; such a constraint is necessary to the existence of a finite delay bound. The exact definitions are recalled in Section \ref{sec:backg:NC}.

The network calculus approach was applied to DRR in \cite{boyer_NC_DRR}, where a \emph{strict} service curve is obtained. 
A strict service curve is a special case of a service curve hence can be used to derive delay (and backlog) bounds. The result was obtained under general assumptions such as per flow maximum packet size and assuming a server that offers any kind of strict service curve rather than a constant-rate server. They show that their delay bounds are smaller than or equal to all previous works \cite{1043123,10.5555/923589,Lenzini_fullexploitation}. We call this the strict service curve of  \emph{Boyer et al.}


The strict service curve of Boyer et al. does not make any assumptions on the interfering traffic. Hence, the resulting delay bounds are valid,  even in degraded operational mode, i.e., when interfering traffic behaves in an unpredictable way. However, in real-time systems,  there is also interest  in finding worst-case delay bounds for non-degraded operational mode, i.e., when interfering traffic behaves as expected and satisfies known arrival curve constraints. 
For such cases, significantly smaller delay bounds were presented at a recent
RTSS conference \cite{Sch_DRR}. The main improvement in \cite{Sch_DRR} is derived as follows.
First, the network calculus delay bound is computed using the strict service curve of Boyer et al.; then, it is improved by what we call the correction term of \emph{Soni et al.} The correction term is obtained by subtracting two terms: The former gives the maximum possible interference caused by any interfering  flow in a backlogged period of the flow of interest and is derived from a detailed analysis of DRR; and the latter gives the effective  interference caused by an interfering flow in a backlogged period of the flow of interest, given the knowledge of an arrival curve of that interfering flow. Unfortunately, the method is semi-rigorous and cannot be fully validated. Indeed, our first contribution is to show that  the correction term of Soni et al. is incorrect; we do so by exhibiting a counter-example that satisfies their assumptions and that has a larger delay (Section~\ref{sec:cunterExm}).

Later, Bouillard, in \cite{anne_drr}, derived new strict service curves for DRR that account for the arrival curve constraints of the interfering traffic and improve on the strict service curve of Boyer et al., hence on the delay bounds. These results require that arrival curves are concave and the aggregate strict service curve (i.e., the strict service curve of the DRR subsystem) is convex.

Our next contribution is obtaining a better strict service curve for DRR when there is no arrival curve constraint on interfering traffic. To do so, we rely on the method of pseudo-inverse, as it enables us to capture all details of DRR; a similar method was used to obtain a strict service curve for Interleaved Weighted Round-Robin in \cite{IWRR_proceeding}. We also provide simplified lower bounds that can be used 
when analytic, closed-form expressions are important. One such lower bound is precisely the strict service curve of Boyer et al. (Fig.~\ref{fig:service}), hence the worst-case delay bounds obtained with our strict service curve are guaranteed to be less than or equal to those of Boyer et al. 

Our following contribution is a new iterative method for obtaining better strict service curves for DRR that account for the arrival curve constraints of interfering flows. Our method is rigorous and is based on pseudo-inverses and output arrival curves of interfering flows. We also provide simpler variants. Our method improves on any available strict service curves for DRR, hence, we always improve on Bouillard's strict service curve. Furthermore, our method accepts any type of arrival  curves, including non-concave ones (such as the stair function used with periodic flows), and can be applied to any type of strict service curve, including non-convex ones (such as the strict service curve we obtained when there is no arrival curve constraint on interfering traffic).

The delay bounds obtained with our method are fully proven. Furthermore, we compute them for the same case studies as in Bouillard's work \cite{anne_drr} (one single server analysis) and as in Soni et al. \cite{Sch_DRR} (including two illustration networks and an industrial-sized one). We find that they are smaller than Bouillard's and the incorrect ones that use the correction term of Soni et al. Hence as of today, it follows that our delay bounds are the best proven delay bounds for DRR, with or without constraints on interfering traffic.


The remainder of the paper is organized as follows. After giving some necessary background in Section~\ref{sec:bg}, we describe the counter example to Soni et al. in \sref{sec:cunterExm}. In Section \ref{sec:serviceCurves}, we present our new strict service curves for DRR, with no knowledge of interfering traffic. In Section \ref{sec:serviceCurveInterfer},
we present our new strict service curves for DRR; they account for the interfering arrival curve constraints. In Section \ref{sec:numEval}, we use numerical examples to illustrate the improvement in delay bounds obtained with our new strict service curves.

This work is the extended version of \cite{Tabatabaee_drr}, which was presented at the RTAS conference, 2021. The conference version did not include a discussion of Bouillard's service curve, which was published after the conference submission date.


\section{Background}
\label{sec:bg}

\subsection{Network Calculus Background}
\label{sec:backg:NC}
We use the framework of network calculus \cite{le_boudec_network_2001, Changbook,bouillard_deterministic_2018}. Let $\mathscr{F}$ denote the set of wide-sense increasing functions
$f:\mathbb{R}^+ \mapsto \mathbb{R^+} \cup \{+\infty\}$. A flow is represented by a cumulative arrival function $A \in \mathscr{F}$
and $A(t)$ is
the number of bits observed on the flow between times $0$ and $t$.
We
say that a flow has $\alpha\in \mathscr{F}$ as \emph{arrival curve} if
for all $s \leq t$,  $A(t) - A(s)\leq \alpha(t-s)$. An arrival curve $\alpha$
can always be assumed to be sub-additive, i.e., to satisfy $\alpha(s+t)\leq \alpha(s)+\alpha(t)$ for all $s,t$. A periodic flow that sends up to $a$ bits every $b$ time units has, as arrival curve, the stair function, defined by $\nu_{a,b}(t)=a\left\lceil \frac{t}{b}\right\rceil$. Another frequently used
arrival curve is the token-bucket function $\alpha=\gamma_{r,b}$, with rate $r$ and burst $b$, defined by $\gamma_{r,b}(t) =
rt+b$ for $t>0$ and $\gamma_{r,b}(t)=0$ for $t= 0$. Both of these arrival curves are sub-additive. 

Consider a system $S$ and a flow through $S$ with input and output functions $A$ and $D$; we say that $S$ offers $\beta\in \mathscr{F}$ as a \emph{strict service curve} to the flow if the number of bits of the flow output by $S$ in any \emph{backlogged} interval $(s,t]$ is $D(t) - D(s)\geq \beta(t-s)$.
%
%
A strict service curve $\beta$ can always be assumed to be super-additive (i.e., to satisfy $\beta(s+t)\geq \beta(s)+\beta(t)$ for all $s,t$) and wide-sense increasing (otherwise, it can be replaced by its super-additive and non-decreasing closure \cite{bouillard_deterministic_2018}). A frequently used strict service curve
is the rate-latency function $\beta_{R,T}\in \mathscr{F}$, with rate $R$ and latency $T$, defined by $
\beta_{R,T}(t) = R[t-T]^+$, where we use the notation $[x]^+=\max\set{x,0}$. It is super-additive.


Assume that a flow, constrained by a sub-additive arrival curve $\alpha$,
traverses a system that offers a strict service curve $\beta$
and that respects the ordering of the flow (per-flow FIFO). The delay of the flow is upper bounded by the horizontal deviation defined by $
h(\alpha,\beta) = \sup_{t \geq 0} \{ \inf \{ d \geq 0 | \alpha(t) \leq \beta(t + d)\}\}
$. 	
Also, the output flow is constrained by an arrival curve $\alpha^* = \alpha \oslash \beta$ where $\oslash$ is the deconvolution operation defined in the next paragraph.
The computation of 	$h(\alpha,\beta) $ and $\alpha^* $ 
can be restricted to $ t \in [0~ t^*]$ for $t^* \geq \inf_{s > 0}\{\alpha(s) \leq \beta(s)\}$ \cite[Prop. 5.13]{bouillard_deterministic_2018}, \cite{stefan}.

For $f$ and $g$ in $\mathscr{F}$, the min-plus convolution is defined by $(f \otimes g)(t) = \inf_{0 \leq s \leq t} \{ f(t-s) + g(s)\}$ and the min-plus deconvolution by $(f \oslash g)(t) = \sup_{s \geq 0} \{ f(t+s) - g(s)\}$
\cite{le_boudec_network_2001,Changbook,bouillard_deterministic_2018}.
We will use the min-plus convolution of a stair function with a linear function, as shown in Fig.~\ref{fig:minplus}.
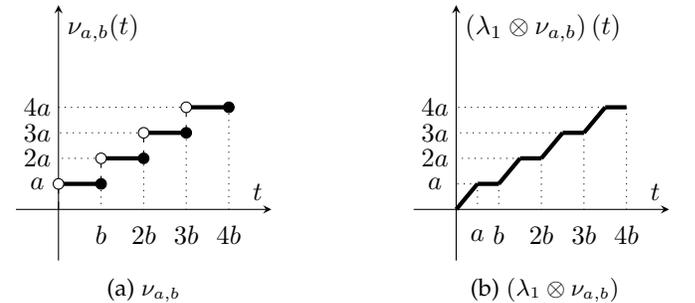
\begin{figure}[htbp]
	\begin{subfigure}[b]{0.4\columnwidth}
		\centering
				\begin{tikzpicture}

		\pgfplotsset{soldot/.style={color=black,only marks,mark=*}} 				   		   		 \pgfplotsset{holdot/.style={color=black,fill=white,only marks,mark=*}}
		\pgfplotsset{ticks=none}
		\begin{axis}[xlabel=$t$, ylabel=$\nu_{a,b}(t)$,
		xmin=-2,xmax= 10,ymin=-2,ymax=8, axis lines=center, width=1.4\textwidth,
		height=1.4\textwidth,]
		
		\addplot[soldot] coordinates{(2,1)(4,2)(6,3)(8,4)};
		\addplot[holdot] coordinates{(0,1)(2,2) (4,3) (6,4)} ;
		
		\draw[dashed] (axis cs:0,0) -- (axis cs:0,1);
		\draw[ultra thick] (axis cs:0,1) -- (axis cs:2,1);
		\draw[dotted] (axis cs:2,1) -- (axis cs:2,0);

		\draw[dashed] (axis cs:2,1) -- (axis cs:2,2);
		\draw[ultra thick] (axis cs:2,2) -- (axis cs:4,2);
		\draw[dotted] (axis cs:4,2) -- (axis cs:4,0);
		\draw[dotted] (axis cs:2,2) -- (axis cs:0,2);
		
		\draw[dashed] (axis cs:4,2) -- (axis cs:4,3);
		\draw[ultra thick] (axis cs:4,3) -- (axis cs:6,3);
		\draw[dotted] (axis cs:6,3) -- (axis cs:6,0);
		\draw[dotted] (axis cs:4,3) -- (axis cs:0,3);
		
		\draw[dashed] (axis cs:6,3) -- (axis cs:6,4);
		\draw[ultra thick] (axis cs:6,4) -- (axis cs:8,4);
		\draw[dotted] (axis cs:8,4) -- (axis cs:8,0);
		\draw[dotted] (axis cs:6,4) -- (axis cs:0,4);

		\draw (2,-1) node{$b$};
		\draw (4,-1) node{$2b$};
		\draw (6,-1) node{$3b$};
		\draw (8,-1) node{$4b$};

		\draw (-1,1) node{$a$};
		\draw (-1,2) node{$2a$};
		\draw (-1,3) node{$3a$};
		\draw (-1,4) node{$4a$};

		\end{axis}
		\end{tikzpicture}
		\caption{$\nu_{a,b}$}
	\end{subfigure}
	\hfill
	\begin{subfigure}[b]{0.4\columnwidth}
		\centering
		\begin{tikzpicture}

		\pgfplotsset{soldot/.style={color=black,only marks,mark=*}} 				   		   \pgfplotsset{holdot/.style={color=black,fill=white,only marks,mark=*}}
		\pgfplotsset{ticks=none}
		\begin{axis}[xlabel=$t$, ylabel=$\left(\lambda_1 \otimes \nu_{a,b}\right)(t) $,
		xmin=-2,xmax=10,ymin=-2,ymax=8, axis lines=center, width=1.4\textwidth,
		height=1.4\textwidth,]

\draw[ultra thick] (axis cs:0,0) -- (axis cs:1,1);
\draw[ultra thick] (axis cs:1,1) -- (axis cs:2,1);
\draw[dotted] (axis cs:1,1) -- (axis cs:0,1);
\draw[dotted] (axis cs:1,1) -- (axis cs:1,0);
\draw[dotted] (axis cs:2,1) -- (axis cs:2,0);

\draw[ultra thick] (axis cs:2,1) -- (axis cs:3,2);
\draw[ultra thick] (axis cs:3,2) -- (axis cs:4,2);
		\draw[dotted] (axis cs:4,2) -- (axis cs:4,0);
\draw[dotted] (axis cs:3,2) -- (axis cs:0,2);

\draw[ultra thick] (axis cs:4,2) -- (axis cs:5,3);
\draw[ultra thick] (axis cs:5,3) -- (axis cs:6,3);
		\draw[dotted] (axis cs:6,3) -- (axis cs:6,0);
\draw[dotted] (axis cs:5,3) -- (axis cs:0,3);

\draw[ultra thick] (axis cs:6,3) -- (axis cs:7,4);
\draw[ultra thick] (axis cs:7,4) -- (axis cs:8,4);
		\draw[dotted] (axis cs:8,4) -- (axis cs:8,0);
\draw[dotted] (axis cs:7,4) -- (axis cs:0,4);
		
		\draw (1,-1) node{$a$};
\draw (2,-1) node{$b$};
\draw (4,-1) node{$2b$};
\draw (6,-1) node{$3b$};
\draw (8,-1) node{$4b$};

\draw (-1,1) node{$a$};
\draw (-1,2) node{$2a$};
\draw (-1,3) node{$3a$};
\draw (-1,4) node{$4a$};

		\end{axis}
		\end{tikzpicture}
		\caption{$(\lambda_1 \otimes \nu_{a,b})$}
	\end{subfigure}
	\caption{\sffamily \small Left: the stair function $\nu_{a,b}\in\mathscr{F}$ defined for $t\geq 0$ by $\nu_{a,b}(t)=a\left\lceil \frac{t}{b}\right\rceil$. Right: min-plus convolution of $\nu_{a,b}$ with the function $\lambda_1\in\mathscr{F}$ defined by $\lambda_1(t)=t$ for $t\geq0$, when $a\leq b$. The discontinuities are smoothed and replaced with a unit slope.}
	\label{fig:minplus}
\end{figure}

If a flow, with arrival curve $\alpha$ and a maximum packet size $\lmax$, arrives on a link with a rate $c$, a better arrival curve for this flow at the output of the link is the min-plus convolution of  $\alpha$ and the function $t\mapsto \lmax + ct$; this is known as grouping (also known as line-shaping) and is also explained Section \ref{sec:soni}.

The non-decreasing closure $f_{\uparrow}$ of a function $f:\mathbb{R}^+ \to \mathbb{R^+} \cup \{+\infty\}$ is the smallest function in $\mathscr{F}$ that upper bounds $f$ and is given by $f_{\uparrow}(t) = \sup_{s \leq t}f(s)$. Also,  the non-decreasing and non-negative closure  $\lb f \rb^+_{\uparrow}$ of $f$ is the smallest non-negative function in $\mathscr{F}$ that upper bounds $f$.

The lower pseudo-inverse $f^{\downarrow}$ of a function $f \in \mathscr{F}$ is defined by
$f^{\downarrow}(y) = \inf \{x | f(x) \geq y \} = \sup \{ x | f(x) < y \}$
and satisfies \cite[Sec. 10.1]{liebeherr2017duality}:
\begin{equation} \label{lem:lsi}
	\forall x,y \in \mathbb{R}^+, y \leq f(x) \Rightarrow x \geq f^{\downarrow}(y)
\end{equation}

The network calculus operations can be automated in tools such as RealTime-at-Work (RTaW) \cite{RTaW-Minplus-Console}, an interpreter that  provides efficient implementations of min-plus convolution, min-plus deconvolution, non-decreasing closure, horizontal deviation, the composition of two functions, and a maximum and minimum of functions for piecewise-linear functions. All computations use infinite precision arithmetic (with rational numbers).

\subsection{Deficit Round-Robin}
\label{sec:drr}

A DRR subsystem serves $n$ inputs, has one queue per input, and uses Algorithm~\ref{alg:DRR} for serving packets. 
Each queue $i$ is assigned a quantum $Q_i$.
DRR runs an infinite loop of \emph{rounds}. In one round, if queue $i$ is non-empty, a service for this queue starts and its  \emph{deficit} is increased by $Q_i$. The service ends when either the deficit is smaller than the head-of-the-line packet or the queue becomes empty. In the latter case, the deficit is set back to zero. The \texttt{send} instruction is assumed to be the only one with a non-null duration. Its actual duration depends on the packet size but also on the amount of service available to the entire DRR subsystem.

\begin{algorithm}[htbp]
	\textbf{Input:}{ Integer quantum $Q_1 , Q_2 ,\ldots, Q_n$}\\
	\textbf{Data:} {Integer deficits: $d_1,d_2,\ldots,d_n$}
	\caption{Deficit Round-Robin}
	\begin{algorithmic}[1]
		\For{$i \leftarrow 1 $ to $n$}
		\State	$d_i \gets 0$;
		\EndFor
		\While{True}
		\For{$i \leftarrow 1 $ to $n$}
		\If{(\textbf{not} empty($i$))} \\	\Comment{A service for queue $i$ starts.} \label{line:start}
		\State	$d_i \gets d_i + Q_i$;
		\While{(\textbf{not} empty($i$))\\ ~~~~~~~~~~~~ and (\texttt{size}(head($i$)) $\leq d_i$)}
		\State$d_i \gets d_i - $\texttt{size}(head($i$);
		\State \texttt{send}(head($i$));
		\State\texttt{removeHead(}$i$);
		\EndWhile \\
		\Comment{A service for queue $i$ ends.} \label{line:end}
		\If{(empty($i$))}
		\State	$d_i \gets 0$;
		\EndIf
		\EndIf
		\EndFor
		\EndWhile
	\end{algorithmic}
	\label{alg:DRR}
\end{algorithm}


In \cite{boyer_NC_DRR} as in much of the literature on DRR, the set of packets that use a given queue is called a \emph{flow}; a flow may however be an aggregate of multiple flows, called micro-flows \cite{charny2000delay} and an aggregate flow is called a \emph{class} in \cite{Sch_DRR}. In this paper, and in order to be consistent with the network calculus conventions, we use the former terminology and consider that a DRR input corresponds to one flow. When comparing our results to \cite{Sch_DRR}, the reader is invited to remember that a DRR flow in this paper corresponds to a DRR class in \cite{Sch_DRR}.

The DRR subsystem is itself placed in a larger system and can compete with other queuing subsystems.
A common case is when the DRR subsystem is at the highest priority on a non-preemptive server with line rate $c$. Due to non-preemption, the service offered to the DRR subsystem might not be instantly available. This can be modelled by means of a rate-latency strict service curve (see Section~\ref{sec:backg:NC} for the definition), with rate $c$ and latency $\frac{c}{L^{\max}}$ where ${L^{\max}}$ is the maximum packet size of lower priority. If the DRR subsystem is not at the highest priority level, this can be modelled with a more complex strict service curve \cite[Section 8.3.2]{bouillard_deterministic_2018}. %
This motivates us to assume that the aggregate of all flows in the DRR subsystem receives a strict service curve $\beta$, which we call ``aggregate strict service curve". If the DRR subsystem has exclusive access to a transmission line of rate $c$, then $\beta(t)=ct$ for $t\geq 0$. We assume that $\beta(t)$ is finite for every (finite) $t$. (Note that the aggregate strict service curve $\beta$  should not be confused with the strict service curves (also called ``residual" strict service curves in \cite{boyer_NC_DRR}) that we obtain in this paper for every flow.)

Here, we use the language of communication networks, but the results equally apply to real-time systems: Simply map flow to task, map packet to job, map packet size to job-execution time, and map strict service curve to ``delivery curve" \cite{4617308,858698}.

\subsection{Strict Service Curve of Boyer et al.} \label{sec:boyer}
The strict service curve of Boyer et al. for  DRR is given in \cite{boyer_NC_DRR}, and we rewrite it using our notation. For flow $i$, let $\mdelta_i$ be its maximum residual deficit, defined by $\mdelta_i= \lmax_i - \epsilon$  where $\lmax_i$ is an upper bound on the packet size and $\epsilon$ is the smallest unit of information seen by the scheduler (e.g., one bit, one byte, one 32-bit word, ...). Also, let $Q_\tot = \sum_{j=1}^{n}Q_j$. Then, for every flow $i$, their strict service curve is the rate-latency service curve $\beta_{R_i, T_i}\lp \beta (t) \rp$ with rate $R_i = \frac{Q_i }{Q_\tot}$ and latency $T_i = \sum_{j \neq i} \mdelta_j + (1 + \frac{\mdelta_i}{Q_i}) \sum_{j \neq i}Q_j$ (see Section~\ref{sec:backg:NC} for the definition of a rate-latency function).
\subsection{Correction Term of Soni et al.} \label{sec:soni}
\begin{figure*}[htbp]
	\centering
	\includegraphics[width=0.9\linewidth]{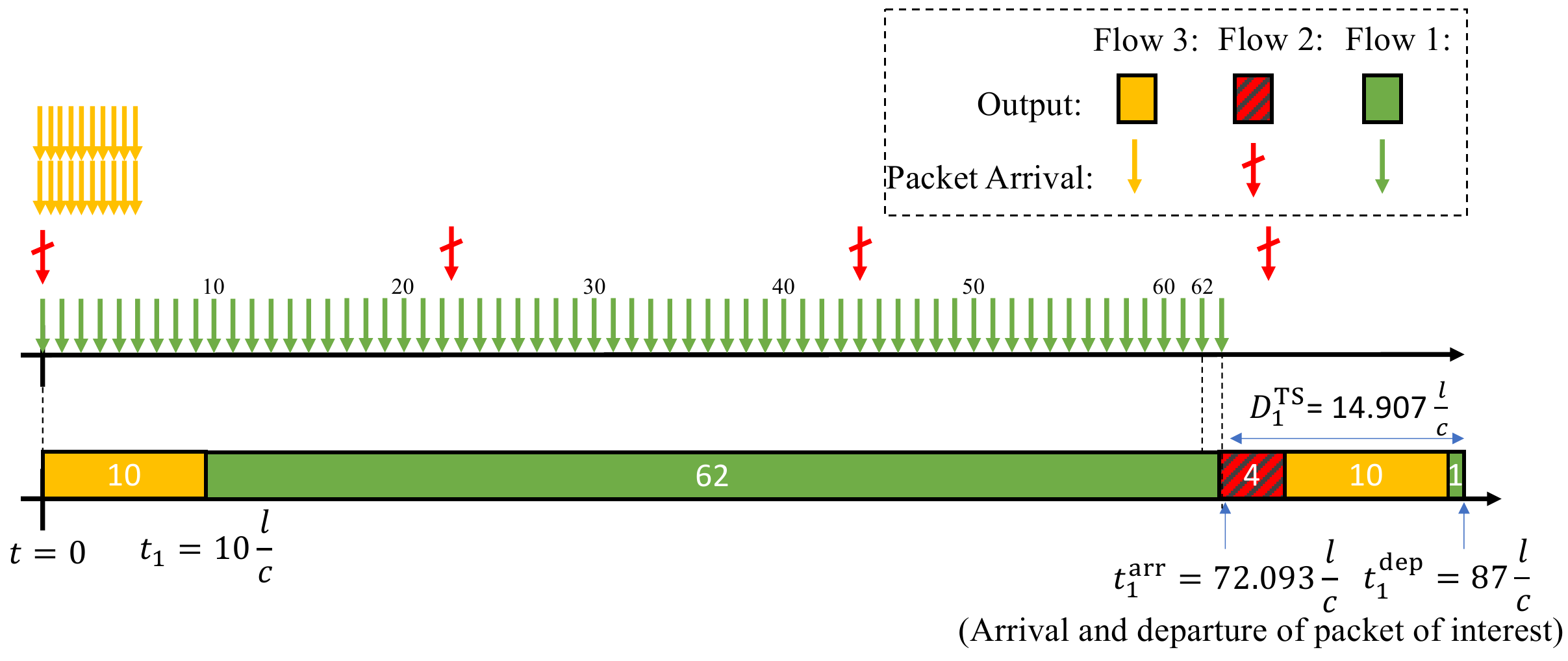}
	\caption{\sffamily \small Trajectory scenario for the packet of interest of flow 1 (Section \ref{sec:TS}). This packet arrives at $\tarr$ and departs at $\tdep$.}
	\label{fig:TL}
\end{figure*}

When interfering flows are constrained by some arrival curves, Soni et al. give a correction term that improves the obtained delay bounds using the strict service curve of Boyer et al. in \cite{Sch_DRR}, which we now rewrite using our notation. Assume that every flow $i$ has an arrival curve $\alpha_i$, and the server is a constant-rate server with a rate equal to $c$. Let $\dboyer_i$ be the network calculus delay bound for flow $i$ obtained by combining $\alpha_i$ with the strict service curve of Boyer et al., as explained in Section \ref{sec:boyer}. The delay bound proposed in \cite{Sch_DRR} is $\dsoni_i = \dboyer_i - \ctsoni_i$ with
\begin{align}\label{eq:soni}
	\ctsoni_i = \sum_{j,j\neq i}\frac{ \lb S_j(\dboyer_i) - \alpha_j(\dboyer_i) \rb^+}{c}
\end{align}
where $
S_j(t) \isdef \lp Q_j + \mdelta_j \rp 1_{t \geq h_i} + Q_j  (1 +  \lfloor \frac{c(t - H_i)}{Q_\tot} \rfloor) 1_{t \geq H_i}  $, $h_i = \frac{ \sum_{j\neq i}Q_j + \mdelta_j }{c}$ and $H_i = h_i + \frac{Q_i - \mdelta_i + \sum_{j\neq i}Q_j }{c}$.
In the correction term $\ctsoni_i$, function $S_j$ represents a lower bound on the maximum interference caused by flow $j$ in a backlogged period of flow $i$; the term with the arrival curve $ \alpha_j$ represents the actual interference caused by flow $j$.

Two additional improvements are used in \cite{Sch_DRR}. The former, called \emph{grouping}, uses the fact that, if a collection of flows is known to arrive on the same link, the rate limitation imposed by the link can be used to derive, for the aggregate flow, an arrival curve that is smaller than the sum of arrival curves of the constituent flows (as explained in Section \ref{sec:backg:NC}). This improvement is also known under the name of line shaping and is used, for example, in \cite{Ahlem-line-shaping,Grieu-line-shaping, bouillard2020tradeoff}. 
The other improvement, called \emph{offsets}, uses the fact that, if several periodic flows have the same source and if their offsets are known, the temporal separation imposed by the offsets can be used to compute, for the aggregate flow, an arrival curve that is also smaller than the sum of arrival curves of the constituent flows (the latter would correspond to an adversarial choice of the offsets). Both improvements reduce the arrival curves, hence the delay bounds. Note that both improvements are independent of the correction term (and, unlike the correction term, are correct); they can be applied to any method used to compute delay bounds, as we do in Section~\ref{sec:numEval}.

\subsection{Bouillard's Strict Service Curves} \label{sec:anne}
A new method to compute strict service curves for DRR that account for the interfering arrival curve constraints was recently presented in \cite{anne_drr}; the method works only when arrival curves are concave and the aggregate strict service curve is convex,  and it improves on the strict service curve of Boyer et al. Specifically, in \cite[Thereom 1]{anne_drr}, for the flow of interest $i$, there exists non-negative numbers $H_J$ for any $J \in \{1,\ldots,n\} \setminus\{i\}$ such that $	\betaAnne_i $ is given by
\begin{equation}
	\betaAnne_i = \sup_{J \subseteq \{1,\ldots,n\} \setminus\{i\}} \frac{Q_i}{\sum_{j \notin  J}Q_j}\lb \beta - \sum_{j \in  J}\alpha_j - \hat{H}_J \rb^+
\end{equation}
where an inductive procedure is presented to compute $\hat{H}_J $. We call these \emph{Bouillard}'s strict service curves.

%
%
%
%

\section{The Correction Term of  Soni et al. is Incorrect}\label{sec:cunterExm}

In this section, we show that the delay bound of Soni et al., namely the correction term given in equation (14) in \cite{Sch_DRR}, is invalid. For flow 1 in a system, we denote the delay bound of Soni et al. by $\dsoni_1$, and we denote the delay experienced by a packet of flow 1 in the trajectory scenario by $\dint_1$.

%

\subsection{System Parameters} \label{sec:sysmodel}
Consider a constant-rate server, with a rate equal to $c$, that uses the DRR scheduling policy. All flows have packets of constant size $l$, and  have quanta $Q_1 = 100l$, $Q_2 = 5l$, and $Q_3 = 10l$.

Each flow is constraint by a token-bucket arrival curve:
\begin{enumerate}
	\item $\alpha_1(t) = \gamma_{r_1,b_1}$ with $r_1 = 0.86 c$ and $b_1 = l$.
	\item $\alpha_2(t) = \gamma_{r_2,b_2}$ with $r_2 = 0.0401c$ and $b_2 = l$.
	\item $\alpha_3(t) =  \gamma_{r_3,b_3}$ with $0 \leq r_3 < \frac{Q_3}{Q_\tot}c $ and $b_3 = 20l$.
\end{enumerate}
Assuming a token-bucket $\gamma_{r,b}$ for a flow implies that this flow has a minimum packet-arrival time equal to $\frac{l}{r}$. Also, observe that $r_i < \frac{Q_i}{Q_\tot}c$ for $i = 1,2,3$.
We compute the delay bound of Soni et al. for flow 1, as explained in \ref{sec:soni}, and we obtain 
$\dsoni_1 = 14.03383 \frac{l}{c} -1.236215\frac{\epsilon}{c}$.

\subsection{Trajectory Scenario
}\label{sec:TS}
We now construct a possible trajectory for our system.
First, we give the inputs of our three flows. All queues are empty, and the server is idle at time $t = 0$. Then,
\begin{enumerate}
	\item Flow 3 arrives first and $A_3(t) = \min\lp  \alpha_3(t), 20l \rp $ for $t >0$ (yellow arrows in Fig.~\ref{fig:TL}).
	\item Flow 1 arrives shortly after flow 3 and $A_1(t) = \min \lp \alpha_1(t), 63l \rp $ for $t >0$ (green arrows in Fig.~\ref{fig:TL}).
	\item Flow 2 arrives shortly after flows 1 and 3 and $A_2(t) = \min \lp \alpha_2(t) , 4l \rp  $ for $t >0$ (red arrows in Fig.~\ref{fig:TL}).
\end{enumerate}
Then, for the output, we have the following:

1) Flow 3 arrives first and has 20 ready packets. As its deficit was zero before this service and $Q_3 = 10l$, the server serves 10 packet of this flow. The end of the service for flow 3 is $t_1 = 10 \frac{l}{c}$ (the first yellow part in Fig.~\ref{fig:TL}).

2) Then, there is an emission opportunity for flow 1 and  $A_1(t_1) = 9.6l$, which means flow 1 has 9 ready packets at time $t_1$. The server starts serving packets of this flow. At the end of service of these first 9 packets, at $t_2 = 19\frac{l}{c}$, flow 1 has another 8 ready packets; hence, the server still serves packets of flow 1. This continues and 62 packets of flow 1 are served in this emission opportunity; the emission opportunity ends at $t_4 = 72 \frac{l}{c}$ (the first green part in Fig.~\ref{fig:TL}).

3) Then, there is an emission opportunity for flow 2 and  $A_2(t_4) = 3.8872l$, which means flow 2 has 3 ready packets at time $t_4$. At the end of service of 3 packets, another packet is also ready for flow 2. In total, 4 packets of flow 2 are served in this emission opportunity (the red part in Fig.~\ref{fig:TL}).

4) A packet for flow 1 arrives at $\tarr = 72 + \frac{0.08l}{r_1} \approx 72.093\frac{l}{c}$. This packet should wait for flow 2 and flow 3 to use their emission opportunities, and then it can be served. We call this \emph{the packet of interest} of flow 1, for which we capture the delay (the first blue arrow, at $\tarr$, on Fig.~\ref{fig:TL}).

5) For flow 3, again 10 packets are served (the second yellow part in Fig.~\ref{fig:TL}).

6) Finally, the packet of interest is served and its departure time is $\tdep = 87 \frac{l}{c}$.

It follows that the delay for the packet of interest is $\dint_1 = \tdep - \tarr = 15\frac{l}{c} - \frac{0.08l}{r_1} \approx  14.907\frac{l}{c}$. Note that $\dint_1 > \dsoni_1$. To fix ideas, if
$l = 100$ bytes and $c = 100$ Mb/s, the delay bounds are $ \dboyer_1 = 146.228 \mu s$, $\dsoni_1 = 112.172 \mu s$, and $ \dint_1= 119.256 \mu s$. 

\subsection{The Contradiction with the Bound of Soni et al.} \label{sec:contradict}

We found a trajectory scenario such that $\dsoni_1$ is not a valid delay bound. Let us explain why the approach of Soni et al., presented in \cite{Sch_DRR},  gives an invalid delay bound.  In \cite{Sch_DRR}, it is implicitly assumed that as the delay for a packet of flow 1 is upper bounded by $ \dboyer_1$ (the obtained delay bound using the  strict service curve of Boyer et al. for flow 1), only packets of interfering flows arriving within a duration $ \dboyer_1$ will get a chance to delay a given packet of flow~1. However, in the trajectory scenario given in Section \ref{sec:TS}, all packets of flow 2 (an interfering flow for flow 1) arriving within the time interval $[0 , 75\frac{l}{c}]$ with the duration $75\frac{l}{c} >>  \dboyer_1 = 18.3 \frac{l}{c}- 2.15\frac{\epsilon}{c}$ delay the packet of interest of flow 1.





	\section{New DRR Strict Service Curve} \label{sec:serviceCurves}
Our next result is a non-convex strict service curve for DRR; we show that it is the largest one and  it dominates the state-of-the-art rate-latency strict service curve for DRR by Boyer et al. We also give simpler, lower approximations of it. Specifically, we also find a convex strict service curve and two rate-latency strict service curves.
\begin{theorem}[Non-convex Strict Service Curve for DRR]\label{thm:drrService}
	Let $S$ be a server shared by $n$ flows that uses DRR, as explained in Section \ref{sec:drr}, with quantum $Q_i$ for flow $i$. Recall that the server offers a strict service curve $\beta$ to the aggregate of the $n$ flows. For any flow $i$, $\mdelta_i$ is the maximum residual deficit (defined in Section~\ref{sec:boyer}).
	
	Then, for every $i$, $S$ offers to flow $i$ a strict service curve $\beta_i^0$ given by $\beta_i^0(t)=\gamma_i \lp \beta(t) \rp$ with
	\begin{align}
		\label{eq:gamma}
		\gamma_i (x) &= \lp \lambda_1 \otimes \nu_{Q_i,Q_{\tot}} \rp \lp  \lb x - \psi_i \lp Q_i - \mdelta_i \rp \rb^+ \rp \\ \nonumber &+ \min  ([x - \sum_{j \neq i} \lp Q_j + \mdelta_j\rp  ] ^+ , Q_i - \mdelta_i  )
		\\
		\label{eq:Qtot}
		Q_{\tot} &= \sum_{j=1}^n Q_j
		\\
		\label{eq:psi}
		\psi_i(x) &\isdef x + \sum_{j,j \neq i} \phi_{i,j} \lp x \rp
		\\
		\label{eq:phi}
		\phi_{i,j}(x) &\isdef \left\lfloor \frac{x + \mdelta_i}{Q_i} \right\rfloor Q_j + \lp  Q_j + \mdelta_j \rp
	\end{align}
	
	Here, $\nu_{a,b}$ is the stair function, $\lambda_1$ is the unit rate function and $\otimes$ is the min-plus convolution, all described in Fig.~\ref{fig:minplus}.
	
	Furthermore, $\beta^0_i$ is super-additive.
\end{theorem}
The proof is in Appendix \ref{sec:proodNonconvex}. See Fig.~\ref{fig:service} for some illustrations of $\beta_i^0$. Observe that $\gamma_i $ in \eqref{eq:gamma} is the strict service curve obtained when the aggregate strict service curve is $\beta=\lambda_1$ (i.e., when the aggregate is served at a constant, unit rate). In the common case where $\beta$ is equal to a rate-latency function, say $\beta_{c,T}$, we have $\beta_i^0(t)=\gamma_i(c(t-T))$ for $t\geq T$ and $\beta_i^0(t)=0$ for $t\leq T$, namely, $\beta_i^0$ is derived from $\gamma_i  $ by a re-scaling of the $x$ axis and a right-shift.

\begin{figure}[htbp]
	\centering
	\input{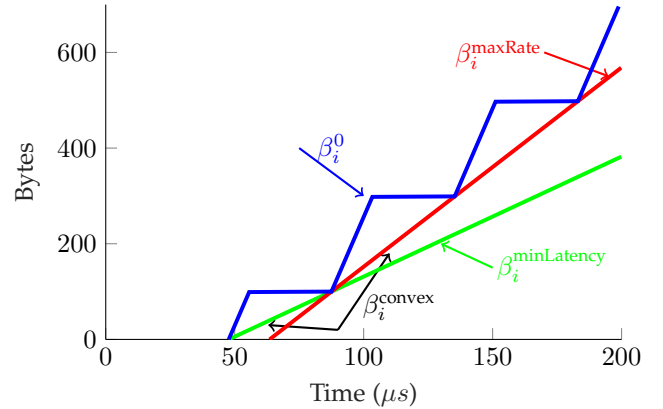}
	\caption{\sffamily \small Strict service curves for DRR for an example with three input flows, quanta $ =\{199, 199, 199\}$ bytes, maximum residual deficits $\mdelta =\{99, 99, 99\}$ bytes, and $\beta(t) = ct$ with $c = 100$ Mb/s (i.e., the aggregate of all flows is served at a constant rate). The figure shows the non-convex DRR strict service curve $\beta_i^0$ of Theorem \ref{thm:drrService}; it also shows the two rate-latency strict service curves $\betaMaxR$ (same as that Boyer et al.) and $\betaMinL$ in Corollary \ref{thm:drrServiceRateLatency} and the convex service curve $\beta_i^{\scriptsize\textrm{convex}} = \max\lp\betaMaxR, \betaMinL\rp$ in Corollary \ref{thm:drrServiceConvex}.
	}
	\label{fig:service}
\end{figure}


We then show that the strict service curve we have obtained is the best possible one.
\begin{theorem}[Tightness of the DRR Strict Service Curve]\label{thm:serviceTight}
	Consider a DRR subsystem that uses the DRR scheduling algorithm, as defined in \sref{sec:drr}. Assume the following system parameters are fixed: the number of input flows $n$, the quantum $Q_j$ allocated to every flow $j$, maximum residual deficits $\mdelta_j$ for every flow $j$, and the strict service curve $\beta$ for the aggregate of all flows. We assume that $\beta$ is Lipschitz-continuous, i.e., there exists a constant $K>0$ such that $\frac{\beta(t)-\beta(s)}{t-s}\leq K$ for all $0\leq s<t$. Let $i$ be the index of one of the flows.
	
	Assume that $b_i\in\mathscr{F}$ is a strict service curve for flow $i$ in any system that satisfies the specifications above. Then $b_i\leq \beta_i^0$ where $\beta_i^0$ is given in Theorem \ref{thm:drrService}.
\end{theorem}

The proof is in Appendix \ref{sec:proofServiceTight}. Note that assuming the aggregate strict service curve $\beta$ is Lipschitz-continuous  does  not  appear  to  be a restriction as the rate at which data is served has a physical limit. We then provide closed-form for the network calculus delay bounds when the flow of interest $i$ is constrained by frequent types of arrival curves, as defined in Section \ref{sec:backg:NC}.

\begin{theorem}[Closed-form Delay Bounds Obtained with the Non-convex Strict Service Curve of DRR] \label{thm:closeDelay}
	Make the same assumptions as in Theorem \ref{thm:drrService}, yet with one difference: Assume that the aggregate strict service curve is a rate-latency function, i.e., $\beta = \beta_{c,T}$. Also, assume that flow of interest $i$ has $\alpha_i \in \mathscr{F}$ as an arrival curve. Let $\psi_i$ be defined as in \eqref{eq:psi}.
	
	Then, the closed-form of the network calculus delay bound $h(\alpha_i, \beta_i^0)$ is given as follows:
	
	1) if $\alpha_i$ is a token-bucket arrival curve, i.e.,  $\alpha_i = \gamma_{r_i,b_i}$ with $r_i \leq \frac{Q_i}{Q_\tot}$,
	\begin{equation}
		T + \max\lp \frac{\psi_i(b_i) }{c},  \frac{\psi_i\lp \alpha_i(\tau_i )\rp }{c}  - \tau_i \rp 
	\end{equation}
	with $	\tau_i  =\frac{Q_i - \lp b_i + \mdelta_i \rp \mymod Q_i }{r_i}$.
	
	2) if $\alpha_i$ is a token-bucket arrival curve and we take into account the effect of grouping, i.e., $\alpha_i(t) = \min \lp ct + \lmax_i , \gamma_{r_i,b_i}(t) \rp $ with $r_i \leq \frac{Q_i}{Q_\tot}$,
	\begin{equation}
		T + \max\lp \frac{\psi_i\lp \alpha_i(\tau_i )\rp  }{c},  \frac{\psi_i\lp \alpha_i(\bar{ \tau}_i ) \rp }{c}  -\bar{ \tau}_i  \rp 
	\end{equation}
	with $\tau_i = \frac{b_i - \lmax_i}{c - r_i}$ and $\bar{ \tau}_i = \frac{Q_i - \lp \alpha_i(\tau_i )  + \mdelta_i \rp \mymod Q_i }{r_i}$.
	
	3) if $\alpha_i$ is a stair arrival curve, i.e.,  $\alpha_i(t) = a_i \lceil \frac{t}{b_i} \rceil$ with $\frac{a_i}{b_i} \leq \frac{Q_i}{Q_\tot}$,
	\begin{equation}
		T + \max\lp \frac{\psi_i(a_i)}{c},  \frac{\psi_i\lp \alpha_i(\tau_i )\rp}{c}  - \tau_i   \rp 
	\end{equation}
	with $  \tau_i   = \lceil\frac{Q_i - \lp a_i + \mdelta_i \rp \mymod Q_i}{a_i} \rceil b_i$
\end{theorem}

The proof is in Appendix \ref{sec:proofDelay}. Theorem \ref{thm:closeDelay} enables us to compute the exact delay bounds in a very simple closed-form, independent of the complicated expression of our non-convex strict service curve. However, we provide simplified lower bounds of the non-convex strict service curve for DRR when analytic, closed-form expressions are important. The function $\phi_{i,j}(x)$, defined in \eqref{eq:phi}, is the maximum interference that flow $j$ can create in any backlogged period of flow $i$, such that flow $i$ receives a service $x$. Using $\phi_{i,j}$ as it is results in the strict service curve of  Theorem \ref{thm:drrService}, which has a complex expression. If there is interest in simpler expressions, any lower bounding function is a strict service curve. In Theorem \ref{thm:drrServiceApp}, we show that any upper bounding of function $\phi_{i,j}$, (which gives a lower bound on $\gamma_i$) results in a lower bound of $\beta_i^0$, which is a valid, though less good, strict service curve for DRR.

\begin{theorem}[Lower Bounds of Non-convex Strict Service Curves for DRR]\label{thm:drrServiceApp}
	Make the same assumptions as in Theorem \ref{thm:drrService}. Also, for flow $i$, consider functions $\phi_{i,j}' \in \mathscr{F}$ such that for every other flow $j \neq i$, $\phi_{i,j}' \geq \phi_{i,j}$. Let $\psi_i'$ be defined as in \eqref{eq:psi}  by replacing functions $\phi_{i,j}$ with $\phi_{i,j}'$ for every flow $j \neq i$. Then, let $\gamma_i'$ be the lower-pseudo inverse of $\psi_i'$, i.e., $\gamma_i' = \psi_i^{'\downarrow}$.
	
	Let $\beta_i^{0'} $ be the result of Theorem \ref{thm:drrService} by replacing functions $\phi_{i,j}$, $\psi_i$, and $\gamma_i$ with $\phi_{i,j}'$, $\psi_i'$, and $\gamma_i'$.
	
	Then, $S$ offers to flow $i$ a strict service curve $\beta_i^{0'}$ and $\beta_i^{0'} \leq \beta_i^0$.
\end{theorem}

The proof is in Appendix \ref{sec:proofApp}. There is often interest in service curves that are piecewise-linear and convex, a simple case is a rate-latency function. Specifically, convex piecewise-linear functions are stable under addition and maximum, and the min-plus convolution can be computed in automatic tools very efficiently 
\cite[Sec.~4.2]{bouillard_deterministic_2018}.
Observe that, if the aggregate service curve $\beta$ is a rate-latency function, replacing $\gamma_i$ by a rate-latency (resp. convex) lower-bounding function also yields a rate-latency (resp. convex) function for $\beta_i^0$, and vice-versa. Therefore, we are interested in rate-latency (resp. convex) functions that lower bound $\gamma_i$. We now give two lower bounds of the non-convex strict service curve for DRR using Theorem \ref{thm:drrServiceApp} that are common: a convex lower bound and two rate-latency lower bounds.

\begin{figure}[htbp]
	\centering
	\input{Figures/phi.tex}
	\caption{\sffamily \small Illustration of functions $\phi_{i,j}$, $\phiMaxR$, $\phiMinL$, and $\phicon$ defined in \eqref{eq:phi}, \eqref{eq:phimaxR}, \eqref{eq:phiminL}, and \eqref{eq:phicon}, respectively. These functions are obtained for the example of Fig.~\ref{fig:service}.}
	\label{fig:phi}
\end{figure}
To obtain a rate-latency strict service curve, we use two affine upper bounds of $\phi_{i,j}$: $\phiMaxR$, which results in a rate-latency function with the maximum rate, and $\phiMinL$, which results in a rate-latency function with the minimum latency (Fig.~\ref{fig:phi}). They are defined by
\begin{align}
	\label{eq:phimaxR}
	\phiMaxR(x) &\isdef \frac{Q_j}{Q_i}\lp x + \mdelta_i  \rp  + Q_j + \mdelta_j
	\\
	\label{eq:phiminL}	
	\phiMinL(x) &\isdef \frac{Q_j}{Q_i - \mdelta_i } x + Q_j + \mdelta_j
\end{align}

\begin{corollary}[Rate-Latency Strict Service Curve for DRR]\label{thm:drrServiceRateLatency}
	With the assumption in Theorem \ref{thm:drrService} and the definitions \eqref{eq:phimaxR}-\eqref{eq:phiminL}, $S$ offers to every flow $i$ strict service curves $\gammaMaxR\lp \beta(t) \rp$  and $\gammaMinL\lp \beta(t) \rp$ with
	\begin{align}
		\label{eq:gammaMaxrate}
		\gammaMaxR &= \beta_{\rmax_i,\tmax_i} \\
		\label{eq:gammaMinlatency}
		\gammaMinL &= \beta_{\rmin_i,\tmin_i}\\	
		\rmax_i &= \frac{Q_i}{Q_\tot} \mand \tmax_i = \sum_{j,j \neq i}\phiMaxR(0)\\
		\rmin_i &= \frac{Q_i - \mdelta_i}{Q_\tot - \mdelta_i} \mand \tmin_i = \sum_{j,j \neq i}\phiMinL(0)
	\end{align}
	The right-hand sides in \eqref{eq:gammaMaxrate} and \eqref{eq:gammaMinlatency} are the rate-latency functions defined in Section~\ref{sec:backg:NC}.
	
\end{corollary}

The above result is obtained by using Theorem \ref{thm:drrServiceApp} with $\phiMaxR$ and $\phiMinL$; hence, $\gamma_i \geq \gammaMaxR$ and $\gamma_i \geq \gammaMinL$. Also, observe that the strict service curve of Boyer et al., explained in Section \ref{sec:boyer}, is equal to $\gammaMaxR\lp \beta(t) \rp $. It follows that $\beta_i^0$ dominates it; hence, obtained delay bound using $\beta_i^0$ are guaranteed to be less than or equal to those of  Boyer et al.



A better upper bound on $\phi_{i,j}$ can be obtained by taking its concave closure (i.e., the smallest concave upper bound) that is equal to the minimum of $\phiMaxR$ and $\phiMinL$:
\begin{equation}\label{eq:phicon}
	\phicon(x) = \min \lp \phiMaxR(x) , \phiMinL(x) \rp
\end{equation}
\begin{corollary}[Convex Strict Service Curve for DRR]\label{thm:drrServiceConvex}
	With the assumption in Theorem \ref{thm:drrService} and the definitions \eqref{eq:gammaMaxrate}-\eqref{eq:gammaMinlatency}, $S$ offers to every flow $i$ a strict service curve $\gammacon \lp \beta(t) \rp$ with
	\begin{align}
		\label{eq:gammaCon}
		\gammacon(x)&= \max\lp\gammaMaxR(x), \gammaMinL(x)\rp
	\end{align}
	
\end{corollary}
The above result is obtained by using Theorem \ref{thm:drrServiceApp} with $\phicon$. Also, it can be shown that it is the largest convex lower bound of $\gamma_i$. When $\beta$ is a rate-latency function, this provides a convex piecewise-linear function, which has all the good properties mentioned earlier.

\section{New DRR strict Service Curves that Account for Arrival Curves of Interfering Flows} \label{sec:serviceCurveInterfer}
The next result provides a method to improve on any strict service curve by taking into account arrival curve constraints of interfering flows. It can thus be applied to the strict service curves presented in Section \ref{sec:serviceCurves} and to Bouillard's strict service curves.
\subsection{A Mapping to Refine Strict Service Curves for DRR by Accounting for Arrival Curves of Interfering Flows}
 \begin{figure*}[htbp]
	\begin{subfigure}[b]{\columnwidth}
		\centering
		\includegraphics[width=\linewidth]{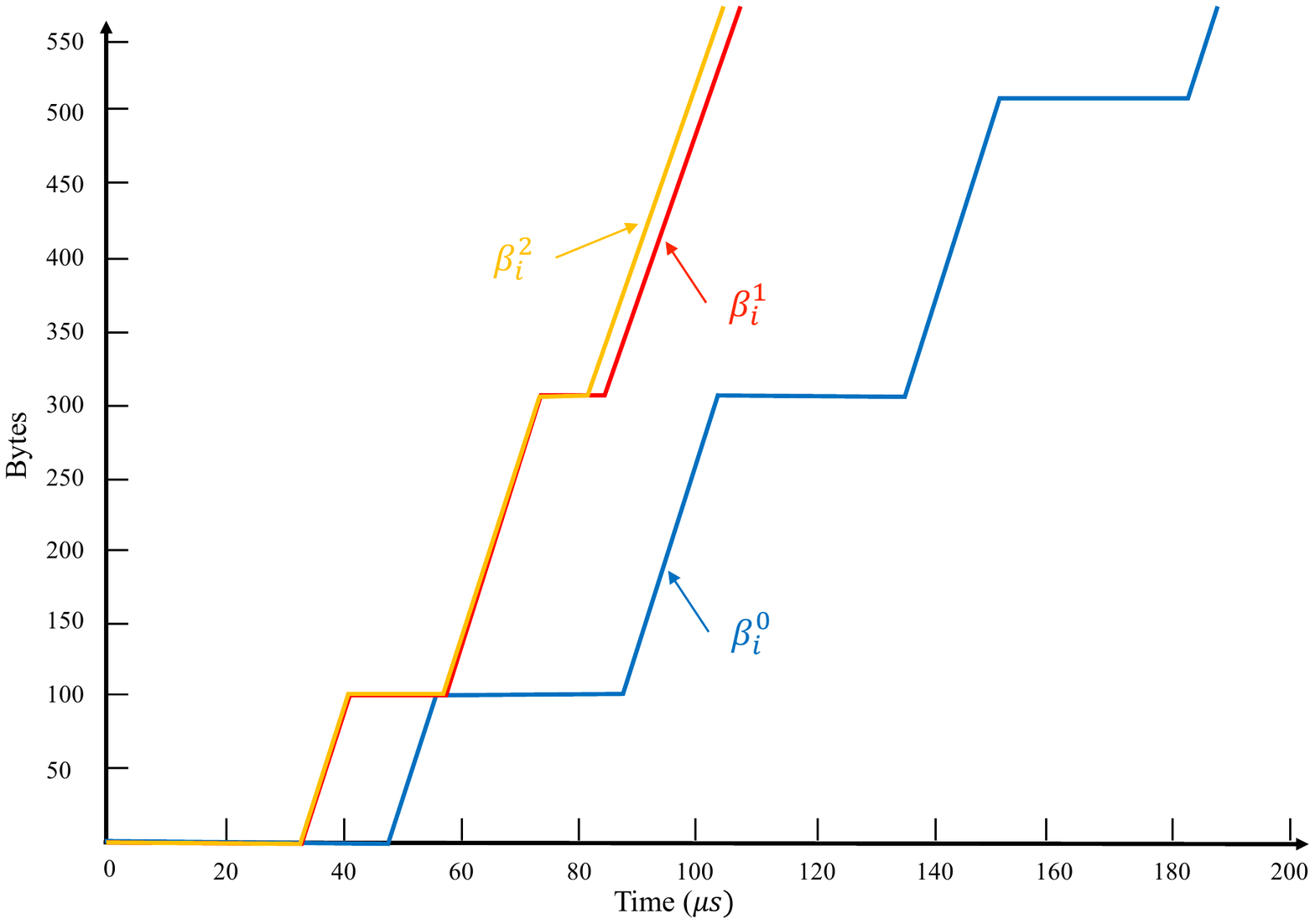}
	\end{subfigure}
	\hfill
	\begin{subfigure}[b]{\columnwidth}
		\centering
		\includegraphics[width=\linewidth]{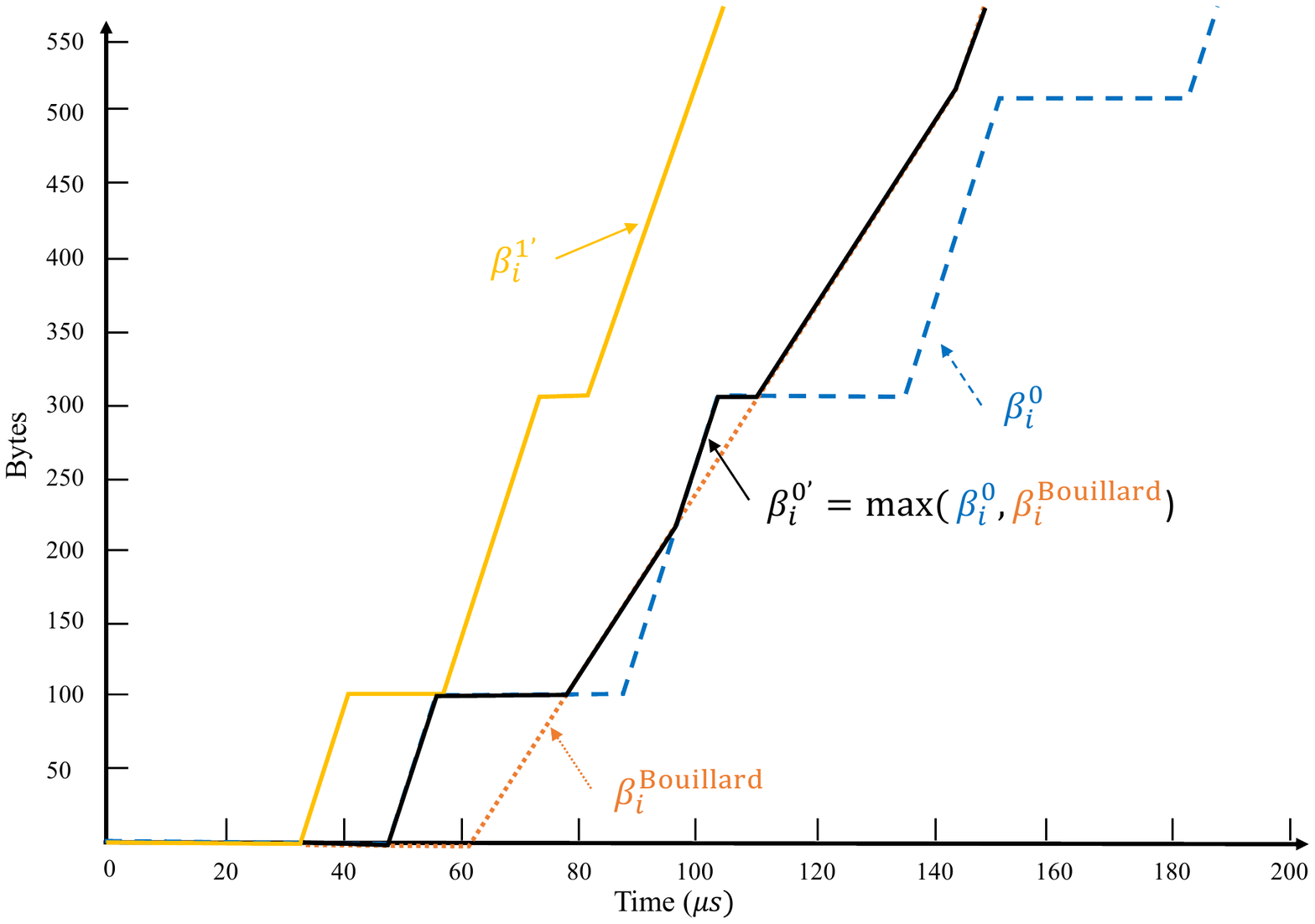}
	\end{subfigure}
	\caption{\sffamily \small  Strict service curves for flow 2 of the example of Fig.~\ref{fig:service}, where all flows have token-bucket arrival curves with $r = \{5,1,1\}\frac{l}{512}$ Mb/s and $b = \{5l,l ,l \}$. When iteratively applying Theorem~\ref{thm:optDrrServiceInd}, starting with either $\beta_i^0$ (Left: the sequence $\beta_i^0 \leq  \beta_i^1 \leq \beta_i^2)$ or starting with $\max \lp \betaAnne_i,\beta_i^0 \rp$ (Right: the sequence ${\beta_i^0}' \leq  {\beta_i^1}' =\beta_i^2$), after $2$ iterations, the strict service curves of all flows become stationary in the horizon of the figure, and the scheme stops  The sufficient horizon $t^*$ in this example is $200 \mu$s. Obtained with the RTaW online tool.}
	\label{fig:IterServiceInd}
\end{figure*}
\begin{figure}[htbp]
	\centering
	\includegraphics[width=\linewidth,height=0.75\linewidth]{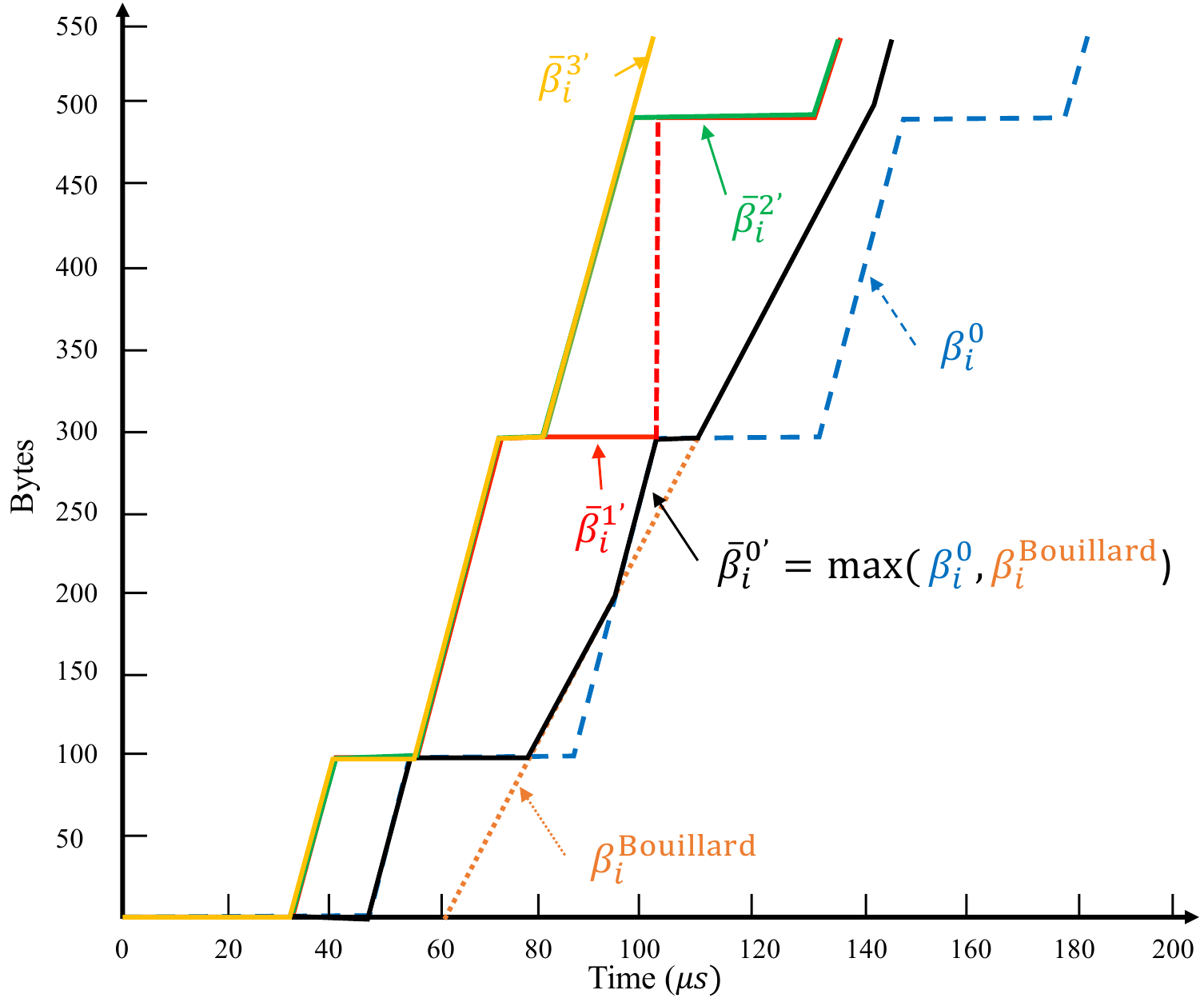}
	\caption{\sffamily \small Strict service curves for flow 2 of the example of Fig.~\ref{fig:IterServiceInd}, when iteratively applying Corollary~\ref{thm:optDrrServiceInd}, starting with either $\beta_i^0$ (Left: the sequence $\bar{\beta}_i^0 \leq  \bar{\beta}_i^1 \leq \bar{\beta}_i^2 \leq \bar{\beta}_i^3)$ or starting with $\max \lp \betaAnne_i,\beta_i^0 \rp$ (Right: the sequence ${\bar{\beta}_i^0}' \leq  {\bar{\beta}_i^1}' \leq {\bar{\beta}_i^2}' \leq {\bar{\beta}_i^3}'$), after $3$ iterations, the strict service curves of all flows become stationary in the horizon of the figure, and the scheme stops.  The sufficient horizon $t^*$ in this example is $200 \mu$s. The strict service curves of both sequences are precisely equal (i.e., $\bar{\beta}_i^1 ={\bar{\beta}_i^1}'$, $\bar{\beta}_i^2 ={\bar{\beta}_i^2}'$, and $\bar{\beta}_i^3={\bar{\beta}_i^3}'$,); also strict service curves of the last step is precisely equal to the last step of Fig.~\ref{fig:IterServiceInd}, i.e., ${\bar{\beta}_i^3}' = \beta_i^2$. Obtained with the RTaW online tool.
	}
	\label{fig:IterServiceRtas}
\end{figure}
\begin{theorem}[Non-convex, Full Mapping] \label{thm:optDrrServiceInd}
	Let $S$ be a server with the assumptions in \thref{thm:drrService}. Also, assume that every flow $i$  has an arrival curve $\alpha_i \in \mathscr{F}$ and a strict service curve $\oldservice_i \in \mathscr{F}$, and let $N_i = \{1,2,\ldots,n\}\setminus\{i\}$, and for any $J \subseteq N_i$, let $\bar{J} =  N_i \setminus J$.
	
	Then, for every flow $i$, a new strict service curve $\newservice_i \in \mathscr{F}$ is given by
	\begin{equation}
		\label{eq:betamJ}
		\newservice_i = \max \lp \oldservice_i, \max_{J \subseteq N_i}\gamma^J_i \circ \lb \beta - \sum_{j \in \bar{J}}\lp \alpha_j \oslash \oldservice_j  \rp \rb_{\uparrow}^+  \rp
	\end{equation}

	%
	with 			
	\begin{align}
		\label{eq:gammaJ}
		\gamma^J_i (x) &= \lp \lambda_1 \otimes \nu_{Q_i,Q_{\tot}^J} \rp \lp  \lb x - \psi^J_i \lp Q_i - \mdelta_i \rp \rb^+ \rp \\ \nonumber &+ \min  ([x - \sum_{j \in  J} \lp Q_j + \mdelta_j\rp  ] ^+ , Q_i - \mdelta_i  )
		\\
		\label{eq:QtotJ}
		Q_{\tot}^J &= \sum_{j \in  J} Q_j
		\\
		\label{eq:psiJ}
		\psi^J_i(x) &\isdef x + \sum_{j \in J} \phi_{i,j} \lp x \rp
	\end{align}
	
	%
	
	In \eqref{eq:betamJ}, $\lb. \rb ^+_{\uparrow}$ is the non-decreasing and non-negative  closure, defined in Section \ref{sec:backg:NC}, and $\circ$ is the composition of functions.
	
\end{theorem}

The proof is in Appendix \ref{sec:proofOptInd}. The essence of Theorem \ref{thm:optDrrServiceInd} is as follows. Equation ~\eqref{eq:betamJ} gives new strict service curves $\newservice_i$ for every flow $i$; they are derived from already available strict service curves $\oldservice_i $ and from arrival curves on the input flows $\alpha_j$; thus, it enables us to improve any collection of  strict service curves that are already obtained.

The computation of service curves in Theorem~\ref{thm:optDrrServiceInd} and of the resulting delay bounds can be restricted to a finite horizon. Indeed, all computations in Theorem~\ref{thm:optDrrServiceInd} are causal except for the min-plus deconvolution $\alpha_j \oslash\oldservice_j$. But, as mentioned in Section~\ref{sec:backg:NC}, such a computation and the computation of delay bounds can be limited to $t\in[0; t^{*}]$ for any positive $t^*$ such that $\alpha_j(t^*)\leq\oldservice_j(t^*)$ for every $m\geq 1$ and $j=1\!:\!n$.
To find such a $t^*$, we can use any lower bound on $\oldservice_j$.

We then compute $t_j^*=\inf_{s > 0}\{\alpha_j(s) \leq \oldservice_j(s)\}$ and take, as \emph{sufficient horizon}, $t^*= \max_j t^*_j$. The computations in Theorem~\ref{thm:optDrrServiceInd} can then be limited to this horizon or any upper bound on it. The computations 
can be performed with a tool such as RealTime-at-Work (RTaW) \cite{RTaW-Minplus-Console}, which uses an exact representation of functions with finite horizon, by means of rational numbers with exact arithmetic.


Theorem \ref{thm:optDrrServiceInd} can be iteratively applied, starting with the strict service curves that do not make any assumptions on interfering traffic, as obtained in Section~\ref{sec:serviceCurves}, and for every flow an increasing sequence of strict service curves is obtained; specifically, recall that $\beta_i^0$ is defined in Theorem \ref{thm:drrService}, then, for every integer $m \geq 1$ and every flow $i$, define $\beta_i^m$  by replacing $\oldservice_j$ with $\beta_j^{m-1}$ in \eqref{eq:betamJ}:
\begin{equation} \label{eq:betaZero}
	\beta_i^m = \max_{J \subseteq N_i}\gamma^J_i \circ \lb \beta - \sum_{j \in \bar{J}}\lp \alpha_j \oslash \beta_j^{m-1}  \rp \rb_{\uparrow}^+
\end{equation}
 It follows that  $\beta_i^m$ is a strict service curve for flow $i$ and $\beta_i^0 \leq \beta_i^1 \leq\beta_i^2 \leq \ldots$ (see Fig.~\ref{fig:IterServiceInd} (left)).

Alternatively, we can first compute Bouillard's strict service curve $\betaAnne_j$ for every flow $j$, as explained in Section \ref{sec:anne}. Observe that $\betaAnne_j$ does not usually dominate the non convex service curve $\beta_i^0$ obtained in Theorem~\ref{thm:drrService} (see Figure~\ref{fig:IterServiceRtas}). Therefore, since the maximum of two strict service curves is a strict service curve, we can take the maximum of both.
%
 Specifically, define ${\beta_i^0}^' = \max \lp \betaAnne_i , \beta_i^0 \rp $, then, for every integer $m \geq 1$ and every flow $i$, define ${\beta_i^m}^'$  by replacing $\oldservice_j$ with ${\beta_j^{m-1}}^'$ in \eqref{eq:betamJ}:
 \begin{equation} \label{eq:betaZeroAnne}
 	{\beta_i^m}^' = \max\lp {\beta_i^{m-1}}^', \max_{J \subseteq N_i}\gamma^J_i \circ \lb \beta - \sum_{j \in \bar{J}}\lp \alpha_j \oslash {\beta_j^{m-1}}^'  \rp \rb_{\uparrow}^+ \rp
 \end{equation}
 It follows that  ${\beta_i^m}^'$ is a strict service curve for flow $i$ and ${\beta_i^0}^'  \leq {\beta_i^{1}}^' \leq {\beta_i^{2}}^' \leq \ldots$ (see Fig.~\ref{fig:IterServiceInd} (right)).

 In practice, in all cases that we tested, when initializing the method with either choice, we always converge to the same strict service curve for every flow (Fig.~\ref{fig:IterServiceInd}).

 We are guaranteed simple convergence for the strict service curves of all flows when iteratively applying Theorem \ref{thm:optDrrServiceInd}, starting from any available strict service curves for all flows. This is because, first, as explained above, computations of such strict service curves can be limited to a sufficient finite horizon; second, by iteratively applying Theorem \ref{thm:optDrrServiceInd}, we obtain an increasing sequence of strict service curves for all flows, and every strict service curve is upper bounded by $\beta$, the aggregate strict service curve. In all cases that we tested, the iterative scheme became stationary in such a finite horizon. Note that the computed strict service curves at each iteration are valid, hence can be used to derive valid delay bounds; this means the iterative scheme can be stopped at any iteration. For example, the iterative scheme can be stopped when the delay bounds of all flows decrease insignificantly.

Observe that the computation to compute strict service curve of Theorem \ref{thm:optDrrServiceInd}, $\newservice_i$ in \eqref{eq:betamJ}, requires $2^{n-1}$ computations of $\gamma^J_i \circ\lb \beta - \sum_{j \in \bar{J}}\lp \alpha_j \oslash \oldservice_j \rp \rb^+_{\uparrow}$ for each $J$ (where $n$ is the total number of  the input flows of the DRR subsystem). In some cases (class based networks) $n$ is small and this is not an issue; in other scenarios (per-flow queuing), this may cause excessive complexity. To address this, we find lower bounds on the strict service curve of Theorem \ref{thm:optDrrServiceInd} where only one computation at each step $m$ is needed; this is less costly when $n$ is large.

\begin{corollary}[Non-convex, Simple Mapping] \label{thm:optDrrService}
	Make the same assumption as in Theorem \ref{thm:optDrrServiceInd}. Then, for every flow $i$, a new strict service curve $\newserviceBar_i \in \mathscr{F}$ is given by
	\begin{equation}
		\label{eq:betam}
		\newserviceBar_i = \max \lp\oldservice_i, \gamma_i \circ\lp \beta + \deltaOld_i \rp_{\uparrow}  \rp
	\end{equation}

	%
	with 			
	\begin{equation}	 			
		\label{eq:delta}
		\deltaOld_i(t)  \isdef\sum_{j,j \neq i} \lb \phi_{i,j} \lp \oldservice_i(t)  \rp - \lp \alpha_j \oslash \oldservice_j \rp (t)\rb^+
	\end{equation}

	Also, $\newserviceBar_i \leq \newservice_i$.
	%
	
	In \eqref{eq:betam}, $_{\uparrow}$ is the non-decreasing closure, defined in Section \ref{sec:backg:NC}, and $\circ$ is the composition of functions; also, note that $\beta$ and $\deltaOld_i$ are functions of the time $t$.
	
\end{corollary}

The proof is in Appendix \ref{sec:proofoptDRR}. The essence of Corollary \ref{thm:optDrrService} is the same as explained after Theorem \ref{thm:optDrrServiceInd}. Corollary \ref{thm:optDrrService}  can be iteratively applied either starting with $\beta_i^0$, defined in Theorem \ref{thm:drrService}, or starting with ${\beta_i^0}^' = \max \lp \betaAnne_i , \beta_i^0 \rp$, i.e., the maximum of $\beta_i^0$ and the strict service curve of Bouillard, explained in Section \ref{sec:anne} (see Fig.~\ref{fig:IterServiceRtas}).

 In the examples that we tested, we observed that the iterative scheme obtained with Theorem \ref{thm:optDrrServiceInd}, our non-convex, full mapping, converges to the same results as the iterative scheme obtained with Corollary \ref{thm:optDrrService}, our non-convex, simple mapping; however, it requires more iterations (see Fig.~\ref{fig:IterServiceInd} and Fig.~\ref{fig:IterServiceRtas}).
  \begin{figure*}[htbp]
 	\begin{subfigure}[b]{\columnwidth}
 		\centering
 		\includegraphics[width=\linewidth]{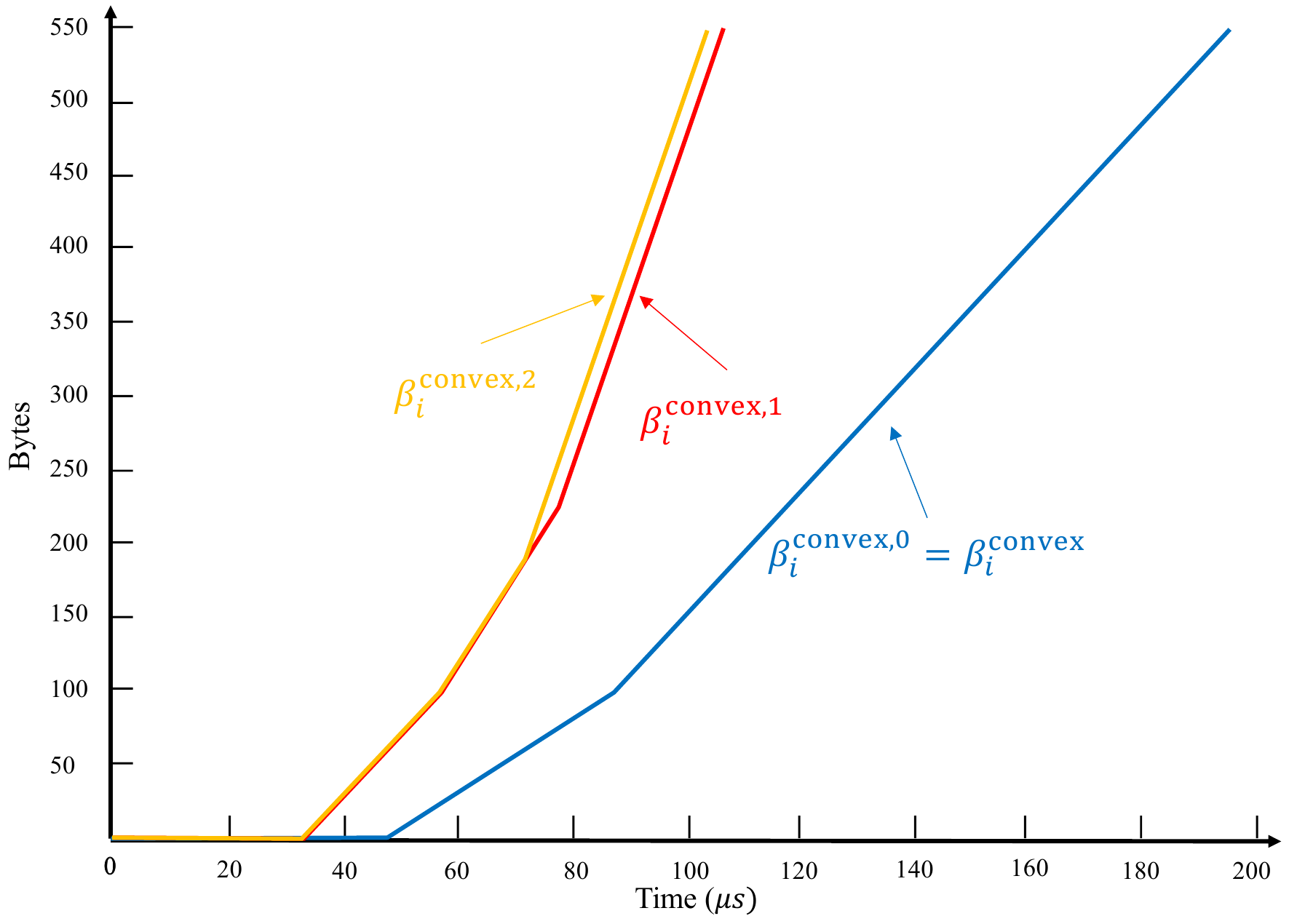}
 	\end{subfigure}
 	\hfill
 	\begin{subfigure}[b]{\columnwidth}
 		\centering
 		\includegraphics[width=\linewidth]{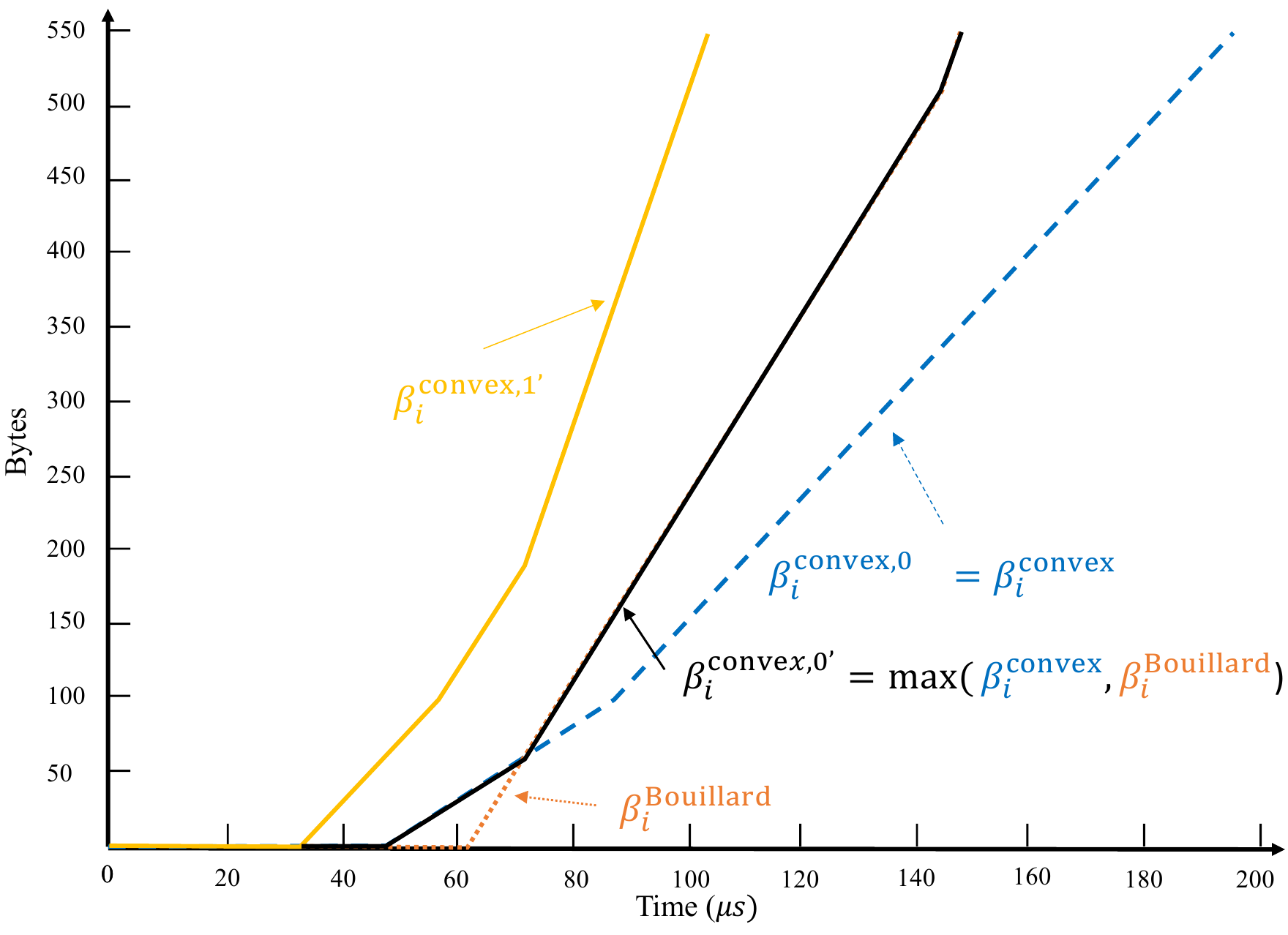}
 	\end{subfigure}
 	\caption{\sffamily \small  Strict service curves for flow 2 of the example of Fig.~\ref{fig:IterServiceInd}, when iteratively applying Corollary~\ref{thm:optDrrServiceInd} as explained in \eqref{eq:betamConJ}, starting with either $\beta_i^{\scriptsize\textrm{convex}}$ (Left: the sequence $\betaconvex{0}_i \leq  \betaconvex{1}_i \leq \betaconvex{2}_i)$ or starting with $\max \lp \betaAnne_i,\beta_i^{\scriptsize\textrm{convex}} \rp$ (Right: the sequence ${\betaconvex{0}_i}'  \leq  {\betaconvex{1}_i}' =\betaconvex{2}_i$), after $2$ iterations, the strict service curves of all flows become stationary in the horizon of the figure, and the scheme stops.  The sufficient horizon $t^*$ in this example is $200 \mu$. Obtained with the RTaW online tool.}
 	\label{fig:IterServiceIndSimpler}
 \end{figure*}

 \subsection{Convex Versions of the Mapping}
 Computation of  the strict service curves of Theorem \ref{thm:optDrrServiceInd} and Corollary \ref{thm:optDrrService} can be costly. We first explain some sources of complexity and how to address them. We then propose  convex versions, for both the non-convex, full  mapping in Theorem \ref{thm:optDrrServiceInd} and the  non-convex, simple mapping in Corollary \ref{thm:optDrrService}.
 \subsubsection{Convex Versions of Theorem \ref{thm:optDrrServiceInd}}

 One source of complexity lies in the initial strict service curves $\beta^0_i$. For every flow $i$, $\beta^0_i$ can be replaced by its simpler lower bounds. As presented in Section \ref{sec:serviceCurves}, $\beta^0_i$ can be replaced by  its convex closure  $\gammacon\lp \beta(t)\rp $, or  rate-latency functions $\gammaMinL\lp \beta(t)\rp $ and  $\gammaMaxR\lp \beta(t)\rp $.

 Another source of complexity is function $\gamma_i^J$, as defined in \eqref{eq:gammaJ}, is non-convex and results in strict service curves that are also non-convex (Fig.~\ref{fig:IterServiceInd}). If there is interest in simpler expressions of Theorem \ref{thm:optDrrServiceInd}, any lower bounding function on $\gamma_i^J$ results in a lower bound of $\newservice_i$, which is a valid, though less good, strict service curve for DRR.

\begin{corollary}[Convex, Full Mapping]\label{thm:optDrrServiceAppInd}
	Make the same assumptions as in Theorem \ref{thm:optDrrServiceInd}. Also, for a flow $i$, let $\hat{\gamma}^{J}_i  \in \mathscr{F}$ such that $\hat{\gamma}^{J}_i \leq  \gamma^{J}_i$.
	
	Let $\newservicehat_i  $ be the result of Theorem \ref{thm:optDrrServiceInd}, in \eqref{eq:betamJ}, by replacing functions $\gamma^J_i$ with $\hat{\gamma}^{J}_i $.
	
	Then, $S$ offers to flow $i$ a strict service curve $\newservicehat_i$ and $\newservicehat_i \leq \newservice_i   $.
\end{corollary}

 As of today, in tools such as RTaW working with functions that are linear and convex is simpler and tractable. Hence, we apply Corollary \ref{thm:optDrrServiceAppInd} with $\hat{\gamma}^{J}_i ={\gammacon}^J = \max\lp {\gammaMaxR}^J , {\gammaMinL}^J\rp$ (convex closure of function $\gamma^{J}_i$) and
\begin{align}
	\label{eq:gammaMaxrateJ}
	{\gammaMaxR}^J &= \beta_{{\rmax_i }^J,{\tmax_i }^J} \\
	\label{eq:gammaMinlatencyJ}
	{\gammaMinL}^J&= \beta_{{\rmin_i}^J,{\tmin_i}^J}\\	
	{\rmax_i }^J&= \frac{Q_i}{Q^J_\tot} \mand {\tmax_i }^J= \sum_{j \in J}\phiMaxR(0)\\
	{\rmin_i}^J &= \frac{Q_i - \mdelta_i}{Q^J_\tot - \mdelta_i} \mand {\tmin_i }^J= \sum_{j \in J}\phiMinL(0)
\end{align}

The sequence of obtained strict service curves  is thus defined by $\betaconvex{0}_i=\gammacon \circ \beta=\beta_i^{\scriptsize\textrm{convex}}$ and for $m\geq 1$ (see Fig.~\ref{fig:IterServiceIndSimpler} (left)):

\begin{equation}
	\label{eq:betamConJ}
	\betaconvex{m}_i = \max_{J \subseteq N_i}{\gammacon}^J\circ \lb \beta - \sum_{j \in \bar{J}}\lp \alpha_j \oslash \betaconvex{m-1}_j \rp \rb^+_{\uparrow}
\end{equation}

Alternatively,  we can first compute the strict service curve of Bouillard $\betaAnne_j$ for every flow $j$, as explained in Section \ref{sec:anne}, and iteratively apply Corollary \ref{thm:optDrrServiceInd} with${\gammacon}^J$, starting with $\max \lp \betaAnne_i , \beta_i^{\scriptsize\textrm{convex}} \rp $; specifically, ${\betaconvex{0}_i}'  = \max \lp \betaAnne_i ,\beta_i^{\scriptsize\textrm{convex}} \rp $, then, for every integer $m \geq 1$ and every flow $i$, define ${\betaconvex{m}_i}' $  as
\begin{equation} \label{eq:betaZeroAnne}
	{\betaconvex{m}_i}' = \max_{J \subseteq N_i}{\gammacon}^J\circ \lb \beta - \sum_{j \in \bar{J}}\lp \alpha_j \oslash {\betaconvex{m-1}_j}' \rp \rb^+_{\uparrow}
\end{equation}
It follows that  ${\betaconvex{m}_i}' $ is a strict service curve for flow $i$ and at each step $m \geq 1$, one can use a better strict service curve $\max \lp {\betaconvex{m}_i}'  , {\betaconvex{m-1}_i}'  \rp$ (see Fig.~\ref{fig:IterServiceIndSimpler} (right)). 

 In practice, in all cases that we tested, when initializing the method with either choice, we always converge to the same strict service curve for every flow (Fig.~\ref{fig:IterServiceIndSimpler}).

Let us explain why computing the above strict service curves is simpler. Min-plus convolution and deconvolution of piecewise linear convex can be computed in automatic tools, such as RTaW, very efficiently \cite[Section 4.2]{bouillard_deterministic_2018}. As illustrated in Fig. \ref{fig:IterServiceIndSimpler}, obtained strict service curves are convex, thus computing the min-plus deconvolution with such strict service curves is much simpler than with those in Fig. \ref{fig:IterServiceInd}. Also, the composition is simpler, as for $f\in\mathscr{F}$, a function $f$, ${\gammacon}^J \lp f(t) \rp$ is equal to $\max \lp {\rmax_i }^J, \lb f(t) - {\tmax_i }^J \rb^+, {\rmin_i}^J\lb f(t) - {\tmin_i}^J\rb^+ \rp$, which includes only multiplication, addition, and maximum operations.

\begin{figure*}[htbp]
	\begin{subfigure}[b]{\columnwidth}
		\centering
		\includegraphics[width=\linewidth]{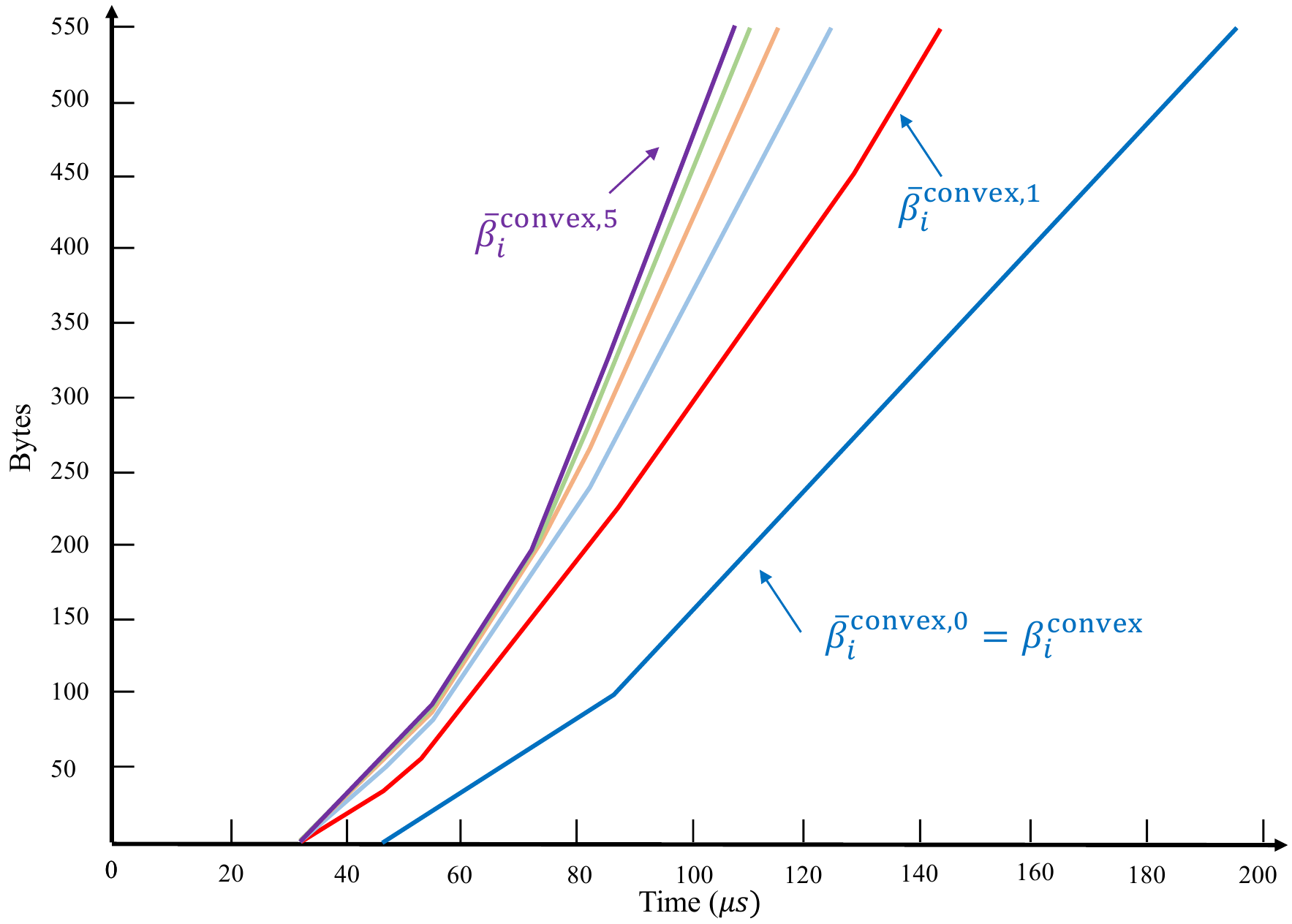}
	\end{subfigure}
	\hfill
	\begin{subfigure}[b]{\columnwidth}
		\centering
		\includegraphics[width=\linewidth]{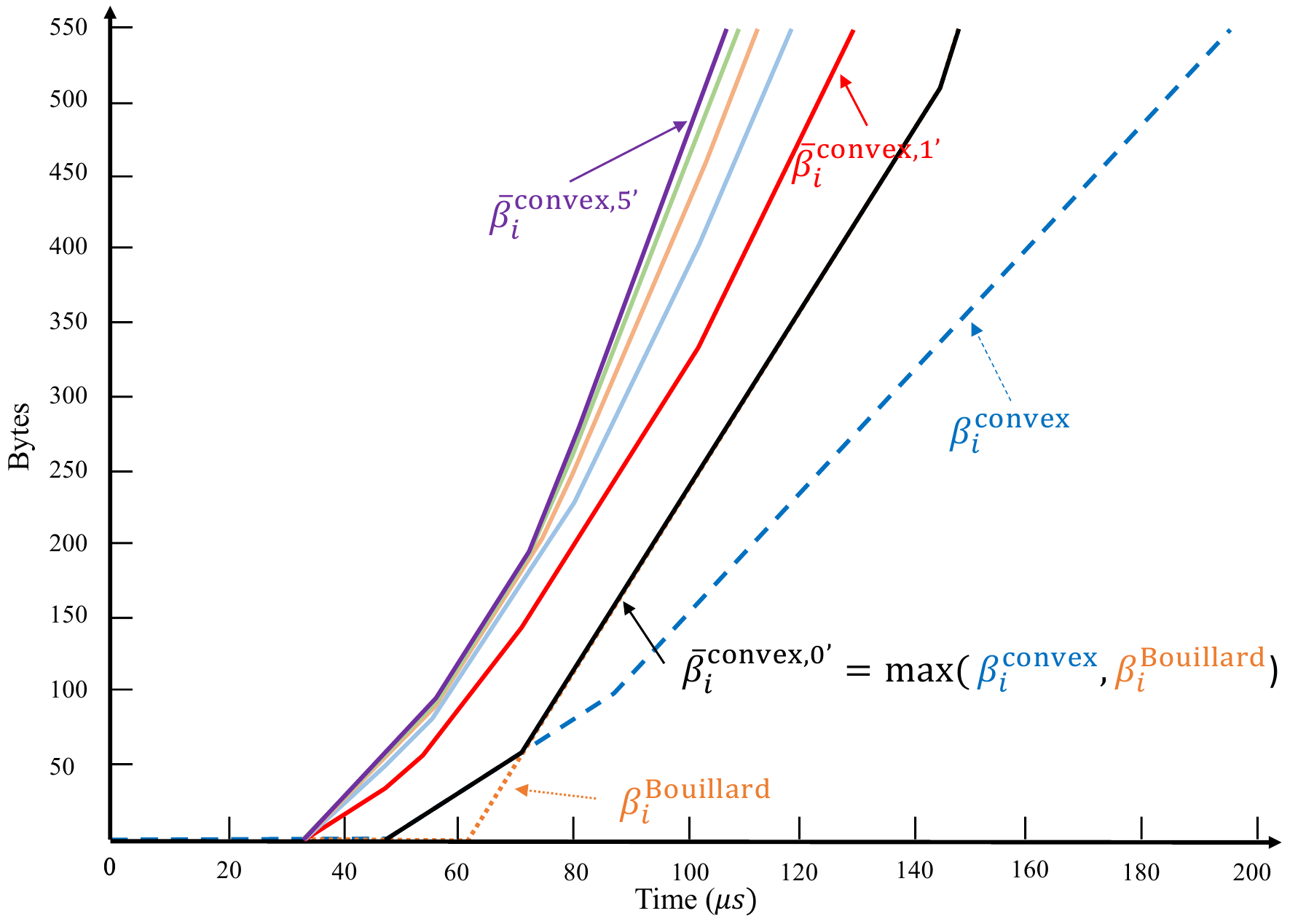}
	\end{subfigure}
	\caption{\sffamily \small Strict service curves for flow 2 of the example of Fig.~\ref{fig:IterServiceInd}, when iteratively applying Corollary~\ref{thm:optDrrServiceApp}, starting with either $\beta_i^{\scriptsize\textrm{convex}}$ (Left: the sequence $\betaaff{0} \leq \betaaff{1}  \leq \ldots)$ or starting with $\max \lp \betaAnne_i,\beta_i^{\scriptsize\textrm{convex}} \rp$ (Right: the sequence ${\betaaff{0}}' \leq {\betaaff{1}}'  \leq \ldots$). The iterative scheme stops when the computed delay bounds for all flows decreased by leas than $0.25\mu$s.  The sufficient horizon $t^*$ in this example is $200 \mu$s. The delay bounds obtained with the strict service curve of the last iteration of both cases are equal, however, the strict service curves are different. The strict service curves of all flows become stationary after 16 iterations. Obtained with the RTaW online tool.}
	\label{fig:IterServiceAff}
\end{figure*}

\begin{figure}[htbp]
	\centering
	\includegraphics[width=\linewidth,height=0.75\linewidth]{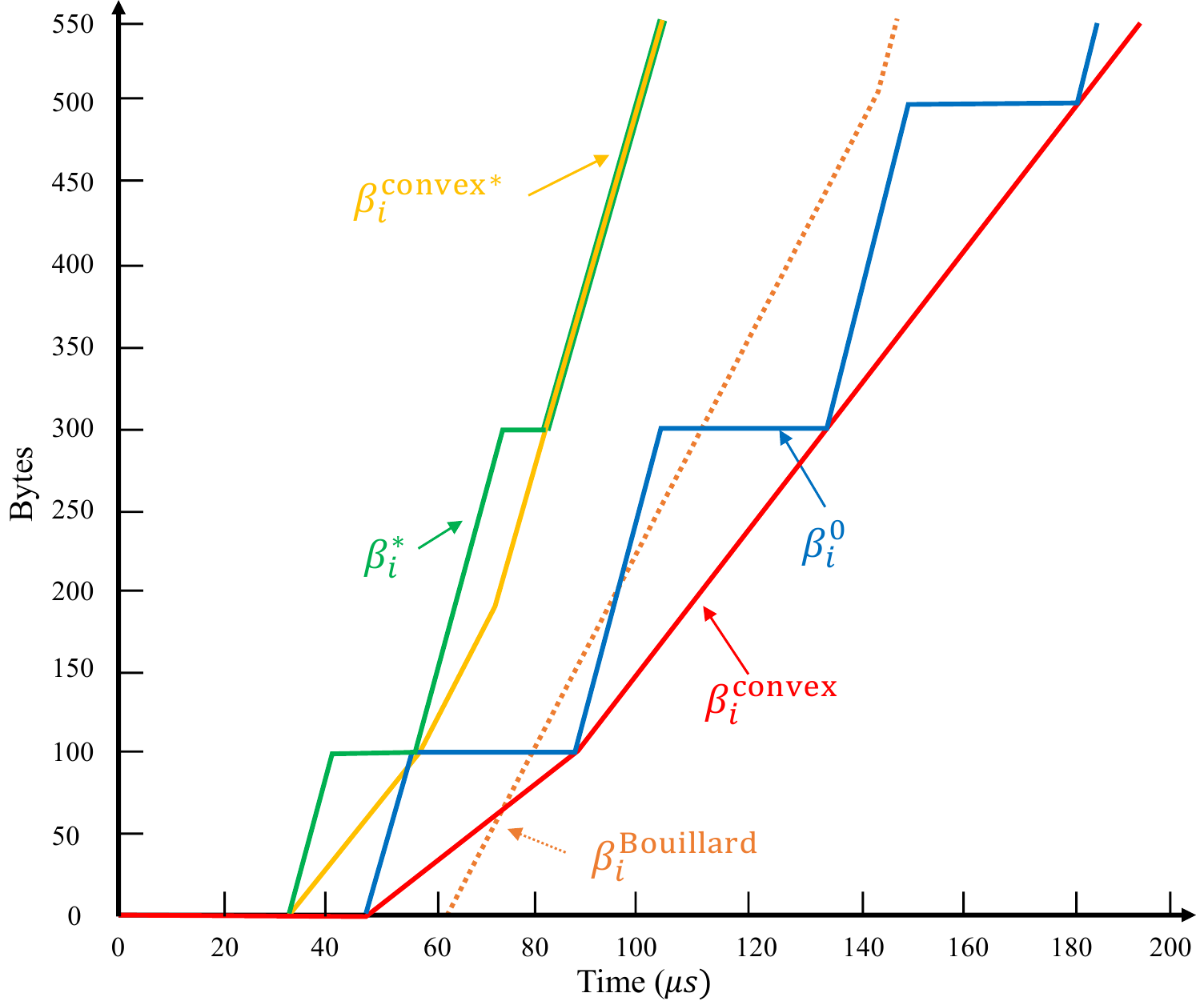}
	\caption{\sffamily \small A summary of strict service curves for flow 2 of the example of Fig.~\ref{fig:IterServiceInd}. The strict service curves $\beta_i^0$ and $\beta_i^{\scriptsize\textrm{convex}}$ are our non-convex and convex strict service curve of Section \ref{sec:serviceCurves}, with no assumption on the interfering traffic.  The strict service curves $\beta_i^*$ and $\beta_i^{\scriptsize\textrm{convex}*}$ are our best non-convex and convex strict service curve that accounts for the interfering traffic, explain in Section \ref{sec:serviceCurveInterfer}. Obtained with the RTaW online tool.
	}
	\label{fig:IterServiceAll}
\end{figure}
 \subsubsection{Convex Versions of Corollary \ref{thm:optDrrService}}

 Again here a source of complexity lies in the initial strict service curves $\beta^0_i$. For every flow $i$, $\beta^0_i$ can be replaced by its simpler lower bounds. As presented in Section \ref{sec:serviceCurves}, $\beta^0_i$ can be replaced by  its convex closure  $\gammacon\lp \beta(t)\rp $, or  rate-latency functions $\gammaMinL\lp \beta(t)\rp $ and  $\gammaMaxR\lp \beta(t)\rp $.

 Also, another source of complexity is function $\phi_{i,j}$ (and the resulting function $\gamma_i$).  Function $\phi_{i,j}$, as defined in \eqref{eq:phi}, is non-concave and non-linear (because it uses floor operations). This might create discontinuities that can make the computation hard, see Fig. \ref{fig:IterServiceRtas}. To address this problem, we 
 derive the following convex version of Corollary \ref{thm:optDrrService}.

\begin{theorem}[Convex, Simple Mapping]\label{thm:optDrrServiceApp}
	Make the same assumptions as in Corollary \ref{thm:optDrrService}. Also, for a flow $i$, let $\phi_{i,j}'$ and $\gamma_i'$ be defined as in Theorem \ref{thm:drrServiceApp}.
	
	Let ${\newserviceBar_i}'$ be the result of Corollary \ref{thm:optDrrService} by replacing functions $\phi_{i,j}$ and $\gamma_i$ with $\phi_{i,j}'$ and $\gamma_i'$, respectively.
	
	Then,  $S$ offers to every flow $i$ a strict service curve ${\newserviceBar_i}'$.
\end{theorem}
The proof is not given in detail, as it is similar to the proof of Corollary \ref{thm:optDrrService} after replacing functions $\phi_{i,j}$ and $\gamma_i$ with $\phi_{i,j}'$ and $\gamma_i'$, respectively.

We apply Theorem \ref{thm:optDrrServiceApp} as follows: Apply Theorem \ref{thm:optDrrServiceApp}  by replacing $\phi_{i,j}$ and $\gamma_i$ with $\phiMaxR$ and $\gammaMaxR$   defined in \eqref{eq:phimaxR} and \eqref{eq:gammaMaxrate}; also,  Apply Theorem \ref{thm:optDrrServiceApp} by replacing $\phi_{i,j}$ and $\gamma_i$ with $\phiMinL$ and $\gammaMinL$   defined in \eqref{eq:phiminL} and \eqref{eq:gammaMinlatency}; then, we take the maximum of the two strict service curves obtained in each case. 

This can be iteratively applied: In both cases, let the initial strict service curves $\betacon{0}_i$  be defined as in Corollary \ref{thm:drrServiceConvex}. Specifically, the sequence of obtained strict service curves are thus defined by either $\betaaff{0}_i=\gammacon \circ \beta= \beta_i^{\scriptsize\textrm{convex}}$ or $\betaaff{0}_i = \max \lp\beta_i^{\scriptsize\textrm{convex}},\betaAnne_i\rp  $ and for $m\geq 1$, $\betaaff{m}_i = \max\lp{\bar{\beta}_i}^{m'},{\bar{\beta}_i}^{m''}\rp$ with 

\begin{equation}
	\label{eq:betamaff}
	\begin{aligned}
	{\bar{\beta}_i}^{m'} & = \gammaMinL \circ \lp \beta + \deltaMinL{m-1}_i \rp_{\uparrow},\\
	{\bar{\beta}_i}^{m''} &=\gammaMaxR \circ\lp \beta+ \deltaMaxR{m-1}_i \rp_{\uparrow}.\\
	\end{aligned}
\end{equation}
Also,  $\deltaMinL{m-1}_i$ and $\deltaMaxR{m-1}_i $ are equal to
\begin{equation}
	\begin{aligned}
		&\sum_{j\neq i} \lb \phiMinL\circ\betaaff{m-1}_i-\alpha_j\oslash\betaaff{m-1}_j \rb^+ \mand\\
		&\sum_{j\neq i} \lb \phiMaxR\circ\betaaff{m-1}_i-\alpha_j\oslash\betaaff{m-1}_j \rb^+,
	\end{aligned}
\end{equation}respectively (see Fig. \ref{fig:IterServiceAff}).

Let us explain why computing the above strict service curves is simpler. The first reason is in computing the composition of $\phiMaxR$ (resp. $\phiMinL$) with another function. Observe that for a function $f\in\mathscr{F}$, $\phiMaxR \lp f(t) \rp$ (resp. $\phiMinL \lp f(t) \rp$) is equal to $ \frac{Q_j}{Q_i}f(t) + \phiMaxR(0)$ (resp. $\frac{Q_j}{Q_i- \mdelta_i}f(t) + \phiMinL(0)$), which includes  only multiplication, addition, and minimum operations. The second reason is in computing the min-plus deconvolution; min-plus convolution and deconvolution of piecewise linear convex can be computed in automatic tools, such as RTaW, very efficiently \cite[Section 4.2]{bouillard_deterministic_2018}, and as illustrated in Fig. \ref{fig:IterServiceAff}, obtained strict service curves are convex, thus computing the min-plus deconvolution with such strict service curves is much simpler than with those in Fig. \ref{fig:IterServiceRtas}. The last reason is in computing the composition of $\gammaMaxR$ (resp. $\gammaMinL$) with another function. Observe that for a function $f\in\mathscr{F}$, $\gammaMaxR \lp f(t) \rp$ (resp. $\gammaMinL \lp f(t) \rp$ ) is equal to $\rmax_i \lb f(t) - \tmax_i \rb^+$ (resp. $\rmin_i \lb f(t) - \tmin_i \rb^+$), which again includes only multiplication, addition, and maximum operations.

Alternatively, one can apply Theorem \ref{thm:optDrrServiceApp}  by replacing $\phi_{i,j}$ and $\gamma_i$ with $\phicon$ and $\gammacon$   defined in \eqref{eq:phicon} and \eqref{eq:gammaCon}; however, in this case, there is no guarantee that this version conserves convexity and we do not consider it further.

 In the examples that we tested, we observed that the iterative scheme obtained with Corollary \ref{thm:optDrrServiceApp}, our convex, full mapping, converges to the same results as the iterative scheme obtained with Theorem \ref{thm:optDrrServiceApp}, our convex, simple mapping; however, it requires more iterations (see Fig.~\ref{fig:IterServiceIndSimpler} and Fig.~\ref{fig:IterServiceAff}).

      \section{Numerical Evaluation} \label{sec:numEval}

In this section, we compare the obtained delay bounds by using our new strict service curves for DRR, presented in Sections \ref{sec:serviceCurves} and \ref{sec:serviceCurveInterfer}, to those of Boyer et al.,  Bouillard, and  Soni et al. We use all network configurations that were presented by Bouillard in \cite{anne_drr} and Soni et al. in  \cite{Sch_DRR}, specifically, one single server, two illustration networks, and an industrial-sized one. For the illustration networks, we use the exact same configuration of flows and switches that Soni et al. use. For the industrial-sized network, Soni kindly replied to our e-mail request by saying that, for confidentiality reasons, they do not have the rights to provide more details about the network configuration than what is already given in \cite{Sch_DRR}. Consequently, we  use the same network but randomly choose the missing information (explained in detail in Section \ref{sec:indus}).


\subsection{Single Server}\label{sec:sibgleServer}
We use the exact  same configuration of flows and the server that Bouillard uses in \cite{anne_drr}. Consider a DRR subsystem with four classes of traffic: Electric protection, Virtual reality games, Video conference, and 4k videos constrained with token bucket arrival curves with bursts $b =\{42.56,2160,3240,7200\}$ kb and rates $r = \{ 8.521, 180, 162, 180\}$ Mb/s, respectively; also, the packet sizes are $\lmax = \{3.04, 12,12,12\}$ kb. The server is a constant-rate server with rate equal to $c=5$Gb/s, i.e., $\beta(t) = ct$. All classes have the same quantum equal to $16000$ bits.

The delay bounds obtained with different methods are given in Table \ref{table:delays}. Our delay bounds always improve on those of Boyer et al.  and Bouillard; when the delay is very small (electric protection) our non-convex service curves bring a considerable improvement. As discussed in Section \ref{sec:serviceCurveInterfer}, the results are the same with the non-convex, full mapping (Theorem~\ref{thm:optDrrServiceInd}) and the non-convex simple mapping (Corollary \ref{thm:optDrrServiceAppInd}). The results of the convex full and simple mappings (Corollary \ref{thm:optDrrService} and Theorem \ref{thm:optDrrServiceApp}) are also identical, but less good than the former.
Also, the results are the same for all choices of initial strict service curves. As we showed that the delay bounds of Soni et al. are incorrect we do not compute them for this example, but for the sake of comparison we will compute them for their case-studies.
\begin{table}
	\centering
		\caption{\sffamily \small Delays bounds of all classes of Section \ref{sec:sibgleServer}.}
	\resizebox{\columnwidth}{!}{
	\begin{tabular} {|c || c c c c c|}
		\hline		
		Class & Boyer et al. & Thm. \ref{thm:drrService} & Bouillard & Thm.  \ref{thm:optDrrServiceInd} & Cor. \ref{thm:optDrrServiceAppInd}\\
		\hline
		Electric protection ($\mu$s) & 52 & 44.51&52 & 44.51 & 52   \\
		\hline
		Virtual reality games (ms) & 1.75 & 1.74&1.33 & 1.32 & 1.32  \\
		\hline
		Video conference (ms)& 2.61& 2.61&1.82 & 1.81 & 1.81   \\
		\hline
		4k videos (ms)& 5.78 & 5.77&2.74 & 2.72 & 2.72 \\
		\hline
	\end{tabular}
}\label{table:delays}
\end{table}

\subsection{Illustration Networks} \label{sec:illus}
\begin{figure*}[htbp]
	\centering
	\includegraphics[width=\linewidth]{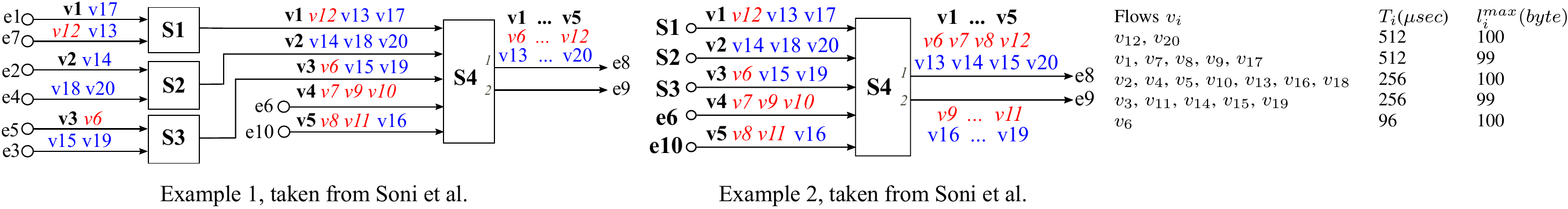}
	\caption{\sffamily \small Networks of Examples 1 and 2, taken from Soni et al \cite{Sch_DRR}. Examples 1 and 2 differ only by the configuration of the switch $S_4$.}
	\label{fig:illus_netwrok}
\end{figure*}

\begin{figure*}[htbp]
	\centering
	\includegraphics[width=\linewidth]{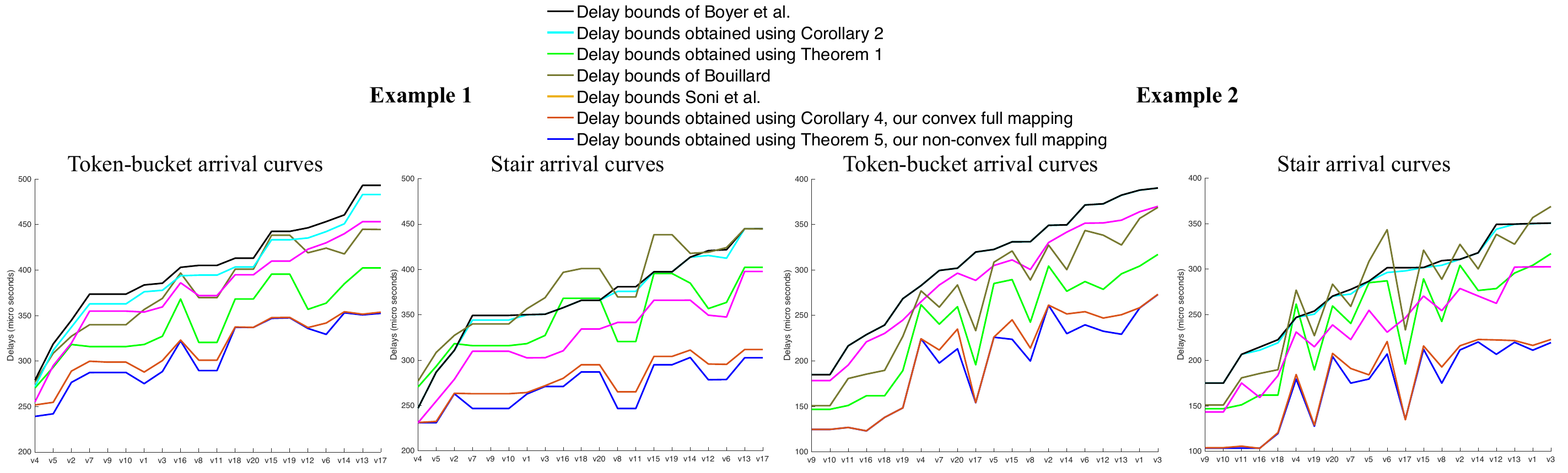}
	\caption{\sffamily \small Delay bounds of flow $v_1, v_2, \ldots, v_{20}$ in Example 1 and Example 2 of Fig.~\ref{fig:illus_netwrok}. In each example, we follow \cite{Sch_DRR} and assume once that flows are constrained by token-bucket arrival curves, and once that flows are constrained by stair arrival curves. The delay bounds of Soni et al. are taken from \cite{Sch_DRR}, and other results are computed with the RTaW online tool. First, delay bounds obtained with our new strict service curves for DRR, with no knowledge on the interfering traffic, are always better than those of Boyer et al. Second, delay bounds obtained with  our new strict service curve for DRR that accounts for arrival curve of interfering flows are always better than those of Soni et al. and  are considerably better than delay bounds of Bouillard. The obtained delay bound obtained with Theorem \ref{thm:optDrrServiceInd} and Corollary \ref{thm:optDrrServiceAppInd}, our non-convex and convex full mapping, are equal to  those of obtained with Corollary \ref{thm:optDrrService} and Theorem \ref{thm:optDrrServiceApp}, our non-convex and convex simple mapping. In each plot, flows are ordered by values of Boyer’s bound.}
	\label{fig:illus}
\end{figure*}


%

Example 1 and 2 are illustrated in Fig.~\ref{fig:illus_netwrok}. We use the exact same network with the exact same configuration for flows and switches as used by Soni et al. in \cite{Sch_DRR}. Examples 1 and 2 differ only by the configuration of the switch $S_4$. Flows $\{v_1 \ldots v_5\}$, $\{v_6 \ldots v_{12}\}$, and $\{v_{13} \ldots v_{20}\}$ are assigned to class $C_1$, $C_2$, and $C_3$, respectively.
There is one DRR scheduler at every switch output port; what we called ``flow" earlier in the paper corresponds here to a class hence $n=3$. Inside a class, arbitration is FIFO (all packets of all flows of a given class are in the same FIFO queue). Also, as in \cite{Sch_DRR}, we assume that queuing is on output ports only.
All classes have the same quantum equal to $199$ bytes. The rate of the links are equal to $c =100$ Mb/s, and every switch $S_i$ has a switching latency equal to $16 \mu$s. Every flow $v_i$ has a maximum packet size $\lmax_i$ and minimum packet arrival $T_i$. Hence, flow $v_i$ is constrained by a token-bucket arrival curve with rate equal to $\frac{\lmax_i}{T_i}$ and burst equal to $\lmax_i$; also, it is constrained by a stair arrival curve given by $\lmax_i\lceil \frac{t}{T_i} \rceil$.

For the sake of comparison, as Soni et al. do not consider grouping and offsets (explained in Section \ref{sec:soni}) in these two examples, we also do not consider them. This means that the arrival curve we use for bounding the input of a class at a switch is simply equal to the sum of arrival curves expressed for every member flow. Arrival curves are propagated using the delay bounds computed at the upstream nodes. We illustrate  the reported values in \cite{Sch_DRR} for the delay bounds of Soni et al. For the other results, we use the RTaW online tool (Fig.~\ref{fig:illus}). As explained in Section \ref{sec:backg:NC}, RTaW provides all the necessary operations to implement our new strict service curves for DRR. First, observe that delay bounds obtained with our new strict service curves for DRR, with no knowledge on the interfering traffic, are always better than those of Boyer et al. Second, delay bounds obtained with our new strict service curve for DRR that accounts for arrival curve of interfering flows are always better than the (incorrect) ones of Soni et al. and are considerably better than Bouillard's. The obtained delay bounds using Theorem \ref{thm:optDrrServiceInd}, our non-convex full mapping, are better than or equal to those of obtained using Corollary \ref{thm:optDrrServiceAppInd}, its convex version; also, they are equal to those of Corollary \ref{thm:optDrrService}, our non-convex simple mapping. Note that the results do not differ whatever the initial strict service curves are. When using token-bucket arrival curves, the run-times (on the RTaW online tool) of Theorem~\ref{thm:optDrrServiceInd} and Corollary \ref{thm:optDrrService} are in the order of $3$ minutes; for their convex versions, Corollary \ref{thm:optDrrServiceAppInd} and Theorem \ref{thm:optDrrServiceApp}, they are in the order of $30$ seconds; when using stair arrival curves, the run-times (on the RTaW online tool) of Theorem \ref{thm:optDrrServiceInd} and Corollary \ref{thm:optDrrService} are in the order of $5$ minutes; for their
 convex versions, Corollary \ref{thm:optDrrServiceAppInd} and Theorem \ref{thm:optDrrServiceApp}, they are in the order of  $1$~minute and $45$~seconds, respectively. 




\subsection{Industrial-Sized Network} \label{sec:indus}
\begin{figure}[htbp]
	\centering
	\includegraphics[width=\linewidth]{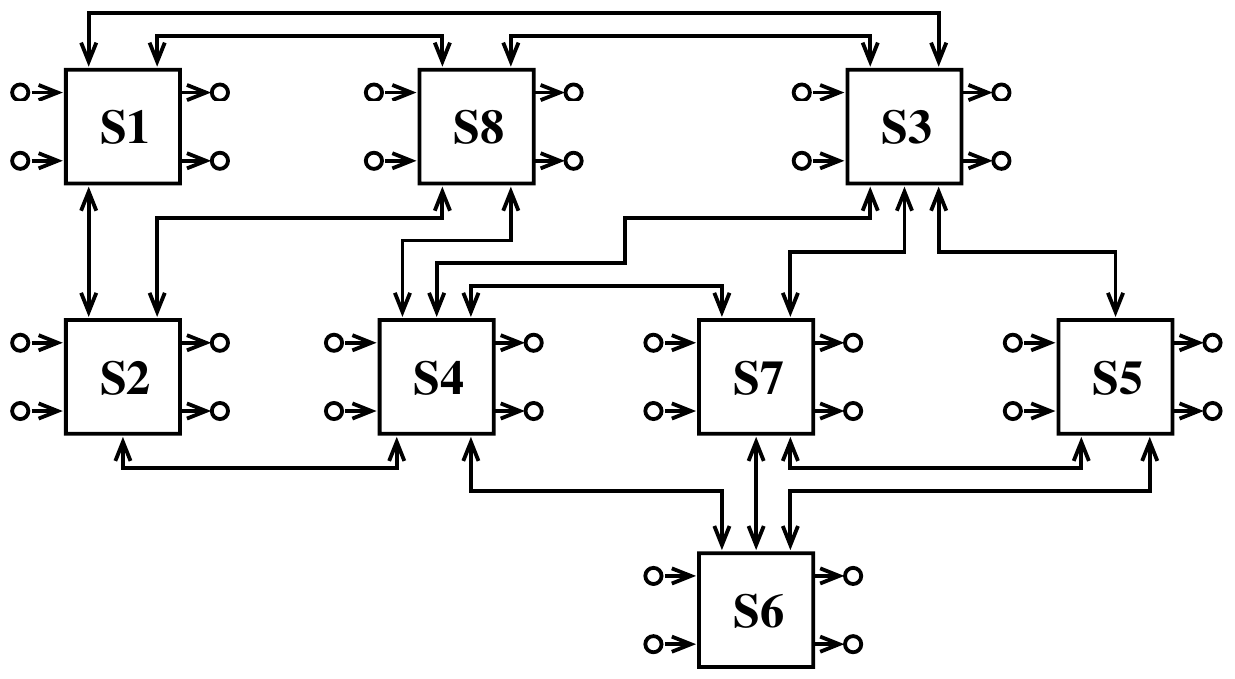}
	\caption{\sffamily \small Industrial-sized network topology. The figure is taken from \cite{1647738}.}
	\label{fig:indusNet}
\end{figure}
\begin{figure}[htbp]
	\centering
	\includegraphics[width=\linewidth]{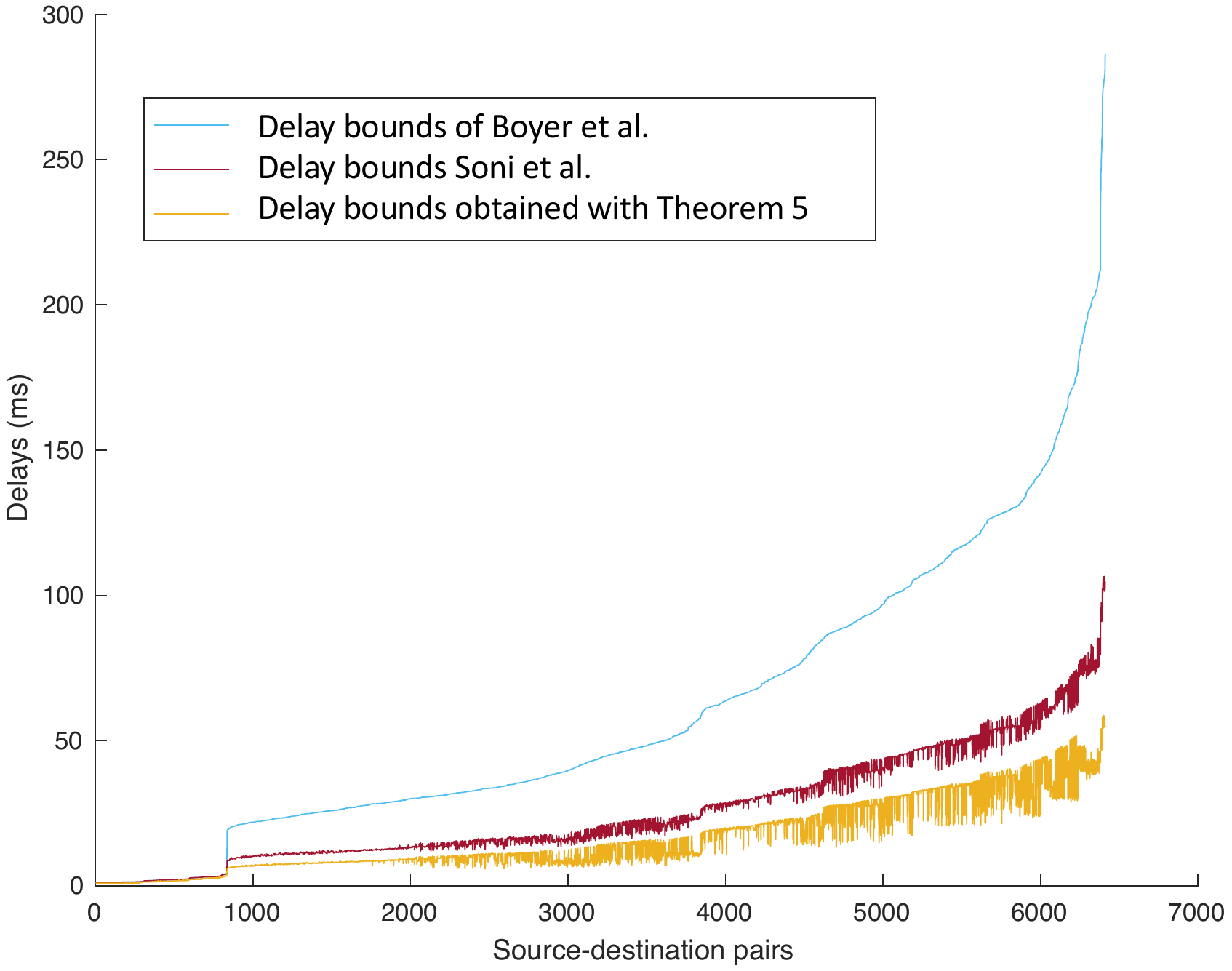}
	\caption{\sffamily \small Delay bounds of the industrial case for all source-destination pairs in the system. The comparison with delay bounds with  Bouillard's method is illustrated in Fig.~\ref{fig:indusAnne}. The obtained delay bound obtained with Theorem \ref{thm:optDrrServiceInd} and Corollary \ref{thm:optDrrServiceAppInd}, our non-convex and convex full mapping, are equal to  those of obtained with Corollary \ref{thm:optDrrService} and Theorem \ref{thm:optDrrServiceApp}, our non-convex and convex simple mapping. Source-destination paths are ordered by values of Boyer’s bound. }
	\label{fig:indusSoni}
\end{figure}
\begin{figure}[htbp]
	\centering
	\includegraphics[width=\linewidth]{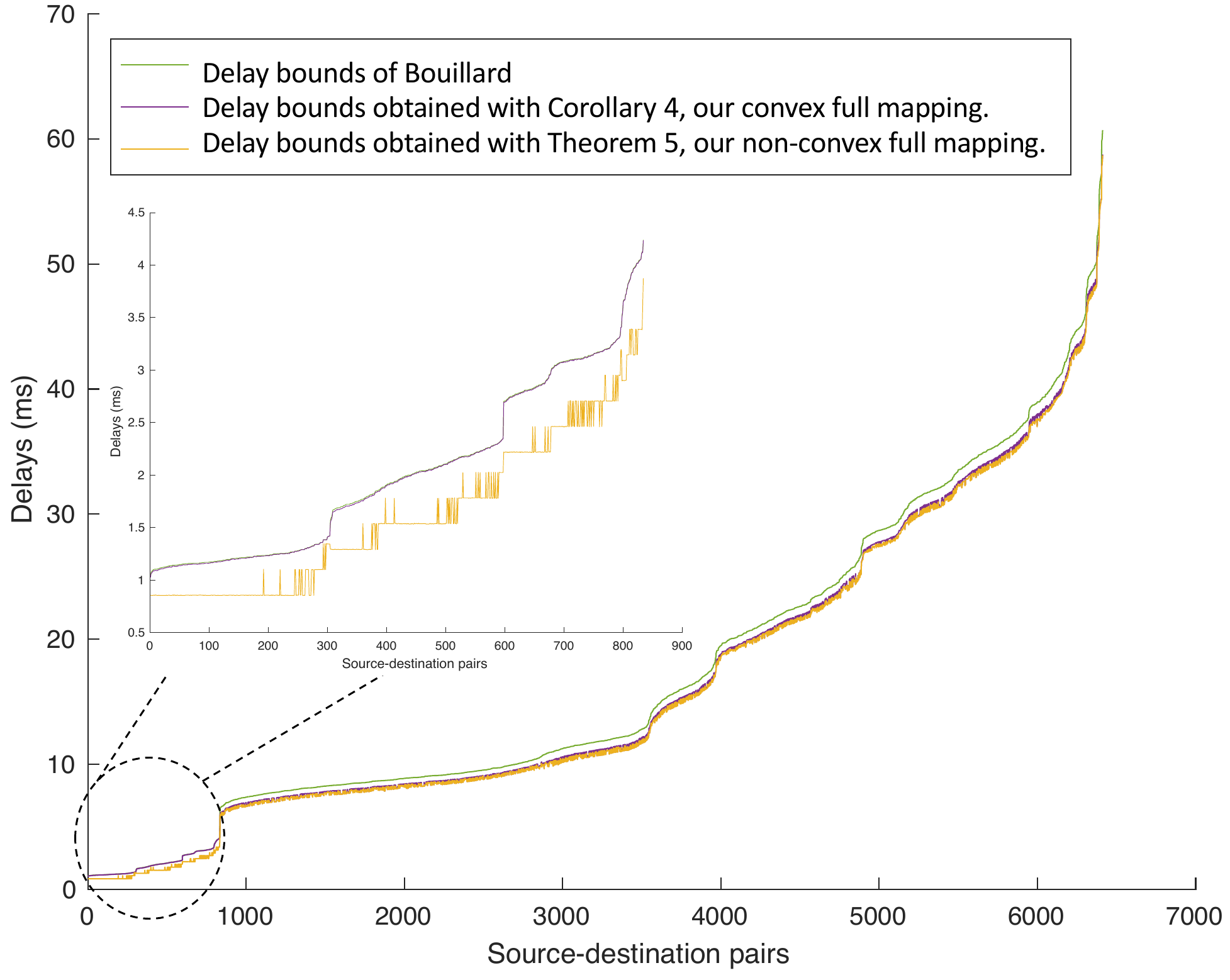}
	\caption{\sffamily \small Delay bounds of the industrial case for all source-destination pairs in the system.  The obtained delay bound obtained with Theorem \ref{thm:optDrrServiceInd} and Corollary \ref{thm:optDrrServiceAppInd}, our non-convex and convex full mapping, are equal to  those of obtained with Corollary \ref{thm:optDrrService} and Theorem \ref{thm:optDrrServiceApp}, our non-convex and convex simple mapping. Source-destination paths are ordered by values of Bouillard’s bound. }
	\label{fig:indusAnne}
\end{figure}




We use the network of Fig.~\ref{fig:indusNet}; it corresponds to a test configuration provided by Airbus in \cite{Grieu-line-shaping}. The industrial-sized case study that Soni et al. use in \cite{Sch_DRR} is based on this network in \cite{1647738}. We combine the available information in both papers to understand this network. It includes 96 end-systems, 8 switches, 984 flows, and 6412 possible paths. The rate of the links are equal to $c =100$ Mb/s, and every switch $S_i$ has a switching latency equal to $16 \mu$s. We find that each switch has 6 input and 6 output end-systems. Three classes of flows are considered: critical flows, multimedia flows, and best-effort flows. There is one DRR scheduler at every switch output port with $n=3$ classes. At every DRR scheduler, the quanta are 3070 bytes for the critical class, 1535~bytes for the multimedia class, and 1535~bytes for the best-effort class.
128 multicast flows, with 834 destinations, are critical; they have a maximum packet-size equal to 150 bytes and their minimum packet arrival time is between $4$ and $128$ms. 500 multicast flows, with 3845 destinations, are multimedia  and their class has a quantum equal to 1535 bytes; they have a maximum packet-size equal to 500 bytes; and their minimum packet-arrival time is between $2$ and $128$ms. 266 multicast flows, with 1733 destinations, are best-effort; they have a maximum packet-size equal to 1535 bytes; and their minimum packet arrival time is between $2$ and $128$ms. For every flow, the path from the source to a destination can traverse at most $4$ switches. Specifically, $1797$, $2787$, $1537$, and $291$ source-destination paths have $1$, $2$, $3$, and $4$ hops, respectively. We choose the paths randomly and satisfy all these constraints.

Due to the limited expressiveness of the language used by the RTaW online tool, we could not implement the industrial-size network there. Therefore, we used MATLAB, which has the required expressiveness. The obtained delay bounds are quasi identical for the full and simple versions of the mappings, therefore we illustrate results only for Theorem \ref{thm:optDrrServiceInd} (non-convex full mapping) and Corollary \ref{thm:optDrrServiceAppInd} (convex full mapping).

Note that the results are identical for both mentioned choices of initial strict service curves.  We also computed the delay bounds obtained with the strict service curve of Boyer et al., with Bouillard's strict service curve  and with the correction term of Soni et al. In all cases, and as in \cite{Sch_DRR}, the arrival curve used for bounding the input of a class at a switch incorporates the effects of delay bounds computed upstream, as well as grouping (line shaping) and offset (Section~\ref{sec:soni}); furthermore, the offsets are such that they create maximum separation, as with \cite{Sch_DRR}. We find that our bounds significantly improve upon the existing bounds, even the incorrect ones (Fig.~\ref{fig:indusSoni}). Moreover, we always improve on Bouillard's delay bounds. Also, delay bounds obtained using Theorem \ref{thm:optDrrServiceInd} are considerably improved compared to its convex version for flows with low delay bounds.


Remark on run-times: For the industrial sized described above,  run-times (on a $2.6$ GHz $6$-Core Intel Core i$7$ computer) of Theorem   \ref{thm:optDrrServiceInd} and its convex version are $96$ and $72$ minutes, respectively; however,  run-times of Corollary   \ref{thm:optDrrService} and its convex version are higher and are $130$ and $103$ minutes, respectively. This is because the number of classes is small, i.e., 3 classes. To increase this, we divided at uniformly random flows of each class to three new classes, which results in 9 classes in total. By doing so, run-times of Theorem   \ref{thm:optDrrServiceInd} and its convex version are $275$ and $220$ minutes, respectively; however,  run-times of Corollary   \ref{thm:optDrrService} and its convex version are lower and are  $162$ and $130$ minutes, respectively. This supports the fact that  computation of strict service curves of Corollary   \ref{thm:optDrrService} is faster than those of Theorem   \ref{thm:optDrrServiceInd}, when the number of flows is large.

	\section{Conclusion}
\label{sec:conc}
The method of the pseudo-inverse enables us to perform a detailed analysis of DRR and to obtain strict service curves that significantly improve the previous results. Our results use the network calculus approach and are mathematically proven, unlike some previous delay bounds that we have proved to be incorrect. Our method assumes that the aggregate service provided to the DRR subsystem is modelled with a strict service curve. Therefore it can be recursively applied to hierarchical DRR schedulers as found, for instance, with class-based queuing. 

	\bibliographystyle{IEEEtran}
\vspace{-0.05in}
\bibliography{ref,leb,boyer}
	\clearpage
\appendices

\twocolumn[
\begin{@twocolumnfalse}
	\begin{center}   
		{\Large\textbf{Supplementary Material}}\\
		{\textbf{Deficit Round-Robin: A Second Network Calculus Analysis}\\
			\textit{Seyed Mohammadhossein Tabatabaee, Jean-Yves Le Boudec}}
	\end{center}
\end{@twocolumnfalse}
]

\section{Proofs} \label{sec:proof}
\subsection{Proof of \thref{thm:drrService}} \label{sec:proodNonconvex}
The idea of the proof is as follows. We consider a backlogged period $(s,t]$ of flow of interest $i$, and we let $p$ be the number of complete service opportunities for flow $i$ in this period, where a complete service opportunity starts at line \ref{line:start} and ends at line \ref{line:end} of Algorithm \ref{alg:DRR}. $p$ is upper bounded by a function of the amount of service received by flow $i$, given in \eqref{eqn:p}. Given this, the amount of service received by every other flow $j$ is upper bounded by a function of the amount of service received by flow $i$, given in \eqref{eq:jksa2}. Using this result gives an implicit inequality for the total amount of service in \eqref{eq:totService}. By using the technique of pseudo-inverse, this inequality is inverted and provides a lower bound for the amount of service received by the flow of interest.


From \cite[Sub-goal 1]{boyer_NC_DRR}, the number $p$ of complete service opportunities for flow of interest, $i$, in $(s,t]$, satisfies
\begin{equation}
	\label{eq:jkhsdef8}
	D_i(t) - D_i(s) \geq pQ_i - \mdelta_i
\end{equation}
Therefore, as $p$ is integer:
\begin{equation}\label{eqn:p}
	p \leq \left \lfloor \frac{D_i(t) -D_i(s) + \mdelta_i}{Q_i} \right \rfloor
\end{equation}
Furthermore, it is shown in the proof of \cite[Sub-goal 2]{boyer_NC_DRR} that
\begin{equation}
	D_j(t) -D_j(s) \leq (p+1)Q_j + \mdelta_j
\end{equation}
Using \eqref{eqn:p} we obtain
\begin{equation}
	D_j(t) -D_j(s) \leq \underbrace{\left \lfloor \frac{D_i(t) - D_i(s) + \mdelta_i}{Q_i} \right \rfloor Q_j  + (Q_j + \mdelta_j)}_{ \phi_{i,j} \lp D_i\lp t \rp - D_i\lp s \rp \rp}
	\label{eq:jksa2}
\end{equation}
Next, as the interval $(s,t]$ is a backlogged period, by the definition of the strict service curve for the aggregate of flows we have
\begin{equation}
	\beta(t-s) \leq (D_i(t) - D_i(s)) + \sum_{j \neq i} \lp D_j(t) - D_j(s) \rp
\end{equation}
We upper bound the amount of service to every other flow $j$ by applying \eqref{eq:jksa2}:
\begin{equation}
	\beta(t-s) \leq \underbrace{(D_i(t) - D_i(s)) + \sum_{j,j \neq i} \phi_{i,j} \left( D_i(t) - D_i(s) \right)}_{\psi_i \lp D_i(t) - D_i(s) \rp }
	\label{eq:totService}
\end{equation}
Then we invert \eqref{eq:totService} using \eqref{lem:lsi} and obtain
\begin{equation}
	D_i(t) - D_i(s) \geq \psi_i^{\downarrow}(\beta(t-s))
\end{equation}	
Lastly, we want to compute $\psi_i^{\downarrow}$. Observe that, by plugging $\phi_{i,j}$ in \eqref{eq:psi},  $\psi_i(x) = x  + \left\lfloor \frac{x + \mdelta_i}{Q_i} \right\rfloor \lp \sum_{j \neq i}Q_j \rp  + \sum_{j \neq i}\lp  Q_j + \mdelta_j \rp$; as there is no plateau in $\psi_i$, its lower-pseudo inverse is simply its inverse which is obtained by flipping the axis (Fig.~\ref{fig:psi_inverse}), and is obtained as
\begin{align} \label{eq:psiinv}
	\psi_i^{\downarrow} (x) = &\lp \lambda_1 \otimes \nu_{Q_i,Q_{\tot}} \rp \lp  \lb x - \psi_i \lp Q_i - \mdelta_i \rp \rb^+ \rp \\ \nonumber&+ \min \lp [x - \sum_{j \neq i} \lp Q_j + \mdelta_j \rp ] ^+ , Q_i - \mdelta_i  \rp
\end{align}
$\psi_i^{\downarrow}$ is illustrated in Fig.~\ref{fig:psi_inverse}. In \eqref{eq:psiinv}, observe that the term with $\min$ expresses the finite part at the beginning between $0$ and $ \psi_i \lp Q_i - \mdelta_i \rp$; also, observe that the term with the min-plus convolution expresses the rest (see Fig.~\ref{fig:minplus}.b with $a = Q_i$ and $b = Q_\tot = \sum_{j=1}^{n} Q_j$.).
\begin{figure}[htbp]
	\centering
	\includegraphics[width=\linewidth]{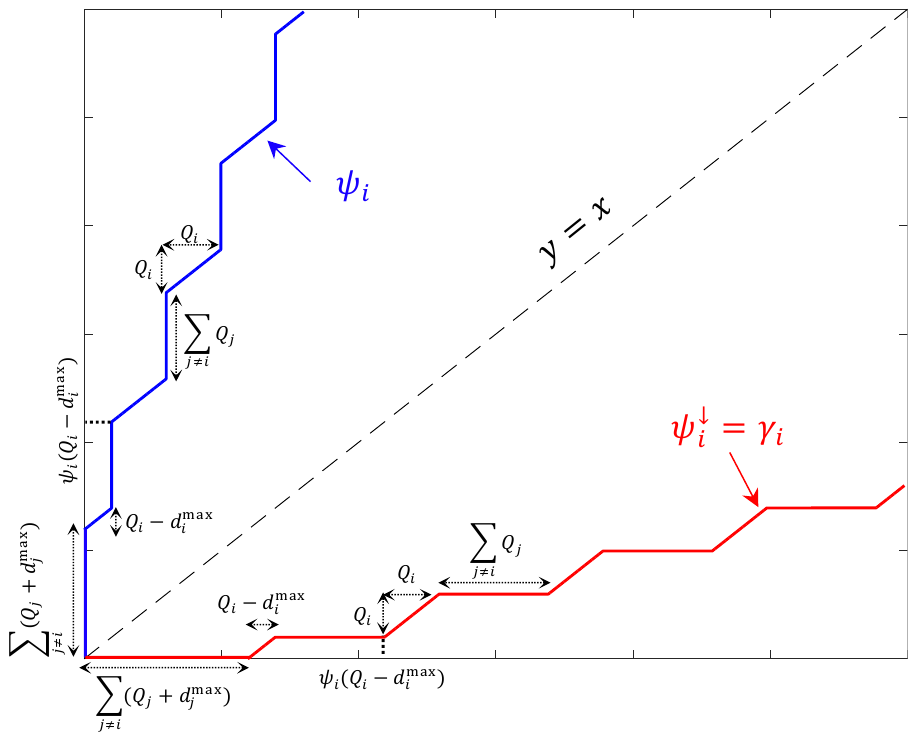}
	\caption{\sffamily \small Illustration of $\psi_i$ and its lower-pseudo inverse $\psi_i^{\downarrow}$, equal to $\gamma_i$, defined in \eqref{eq:psi} and \eqref{eq:psiinv}, respectively. Function $\gamma_i^J$ has the same form as $\gamma_i$.  
	}
	\label{fig:psi_inverse}
\end{figure}
\qed

\subsection{Proof of Theorem \ref{thm:serviceTight}} \label{sec:proofServiceTight}

We use the following Lemma about the lower pseudo-inverse technique.

\begin{lemma}[Lemma 14 \cite{tabatabaee2020interleaved}] \label{lsi:lem}
	For a right-continuous function $f$ in $\mathscr{F}$ and $x,y$ in $\mathbb{R^+}$, $f^{\downarrow}\left(y \right) = x$ if and only if
	$
	f(x) \geq y $ and there exists some $\varepsilon>0$ such that $ \forall x' \in (x-\varepsilon, x), f(x') < y
	$.
\end{lemma}

%
%
%
%
%


\begin{proof}[Proof of Theorem \ref{thm:serviceTight}]
	\begin{figure*}[htbp]
		\centering
		\includegraphics[width=0.9\linewidth]{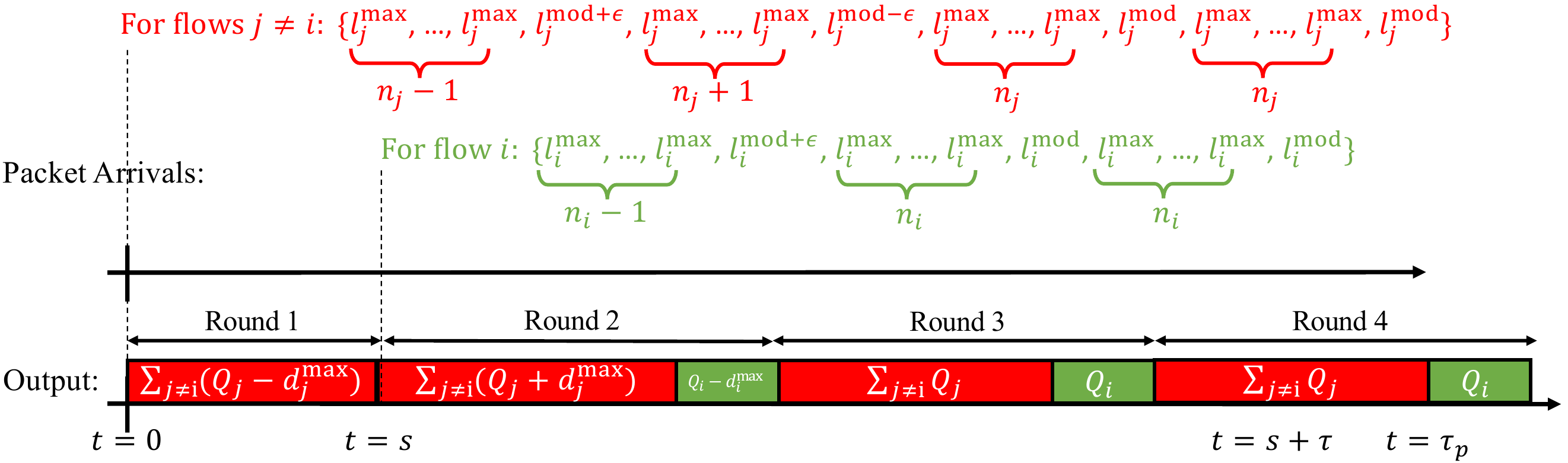}
		\caption{\sffamily \small Example of the trajectory scenario presented in Section \ref{sec:proofServiceTight} with $p = 2$.}
		\label{fig:tight}
	\end{figure*}
	%
	We prove that, for any value of the system parameters, for any $\tau >0$, and for any flow $i$, there exists one trajectory of a system such that
	\begin{equation}\label{eqn;tight}
		\begin{aligned}
			&\exists s \geq 0, \, (s,s+\tau] \, \text{ is backlogged for flow $i$}\\ &\text{and }D_i(s + \tau) - D_i(s) = \beta_i^0(\tau)
		\end{aligned}
	\end{equation}
	
	\textbf{Step 1: Constructing the Trajectory}
	
	1) We use the following packet lengths for each flow $j$:  As explained in Section \ref{sec:drr}, we have $Q_j \geq \lmax_j$; thus, there exist an integer $n_j \geq 1$ such that $Q_j = n_j \lmax_j + (Q_j \mymod \lmax_j)$. Then, let $\lmod_j = (Q_j \mymod \lmax_j)$, $\lmodm_j = \lmod_j - \epsilon $, $\lmodp_j = \lmod_j + \epsilon $. 
	
	2) Flows are labeled in order of quanta, i.e., $Q_j\leq Q_{j+1}$.
	
	3) At time $0$, the server is idle, and  the input of every queue $j \neq i$ is a bursty sequence of packets as follows:
	\begin{itemize}
		\item First, $(n_j -1)$ packets of length $\lmax_j$ followed by a packet of length $\lmodp_j$; 
		\item Second, $(n_j +1)$ packets of length $\lmax_j$ followed by a packet of length $\lmodm_j$. Note that if $\lmod_j=0$, the sequence can be changed to $n_j$ packets of length $\lmax_j$ followed by a packet of length $\lmax_j - \epsilon$ and the rest of the proof remains the same; 
		\item Third, $\lp \lfloor \frac{\beta_i^0(\tau) + \mdelta_i}{Q_i} \rfloor - 1 \rp $ times of a sequence of $n_j $ packets of length $\lmax_j$ followed by a packet of length $\lmod_j$.
	\end{itemize}
	
	4) Let flow $i'$ be the first flow that is visited after flow $i$ by the DRR subsystem, i.e., $i'= (i+1) \mymod n $. The input of flow $i'$ arrives shortly before all other flows $j \neq i$ at time 0.
	
	5) The output of the system is at rate $K$ (the Lipschitz constant of $\beta$) from time $0$ to times $s$, which is defined as the time at which queue $i$ is visited in the second round, namely
	\begin{equation}\label{eq-kjlsda}
		s=\frac{1}{K}\sum_{j\neq i} \lp Q_j - \mdelta_j \rp 
	\end{equation}
	It follows that  \begin{equation}\label{eq-asfdlsda}\forall t\in[0,s], D(t)=Kt\end{equation}
	
	
	6) The input of queue $i$ starts just after time $s$, with a bursty sequence of packets as follows:
	\begin{itemize}
		\item First, $(n_i -1)$ packets of length $\lmax_i$ followed by a packet of length $\lmodp_i$; 
		\item Second, $\lp \lfloor \frac{\beta_i^0(\tau) + \mdelta_i}{Q_i} \rfloor \rp $ times of a sequence of $n_i $ packets of length $\lmax_i$ followed by a packet of length $\lmod_i$.
	\end{itemize}
	
	7) After time $s$, the output of the system is equal to the guaranteed service; by 3) and 6), the busy period lasts for at least $\tau$, i.e.,
	\begin{equation}\label{eq-jhgdsf}
		\forall t\in [s,s+\tau], D(t)=D(s)+\beta(t-s)
	\end{equation}  In particular,
	\begin{equation}
		\label{eqn:3}
		D(s + \tau) - D(s) = \beta(\tau)
	\end{equation}
	
	\textbf{Step 2: Analyzing the Trajectory}
	
	Let $p$ be the number of complete services for flow $i$ in $(s, s + \tau ]$, and let $\tau_p$ be the start of the first service for flow $i$ after these $p$ services. We want to prove that 
	\begin{equation}\label{eq:phiTight}
		D_j(\tau_p) - D_j(s) = \phi_{i,j} \lp D_i(s + \tau) - D_i(s) \rp
	\end{equation}
	
	We first analyze the service received by every other flow $j \neq i$. First, observe that every $j \neq i$ sends $(n_j -1)$ packets of length $\lmax_j$ followed by a packet of length $\lmodp_j$ in the first service after $t = 0$; this is because at the end of serving these packets the deficits of flow $j$ becomes $\mdelta_j$ and the head-of-the-line packet has a length $\lmax_j > \mdelta_j$. Second, for the first service after time $s$,  every other flow $j \neq i$ sends $(n_j  + 1)$ packets of length $\lmax_j$ followed by a packet of length $\lmodm_j$, and at the end of this service the deficit becomes zero. Third, observe that in any other complete services for flow $j$ (if any), it sends $n_j $ packets of length $\lmax_j$ followed by a packet of length $\lmod_j$. Hence, in the first complete service of flow $j$ after time $s$, flow $j$ is served by $\lp Q_j + \mdelta_j \rp$; and in every other complete services for flow $j$, it is served by $Q_j$. (the red parts in Fig.~\ref{fig:tight})
	
	We then analyze the  service received by flow $i$. First, it should wait for all other $j\neq i$ to use their first service after time $s$, and then flow $i$ sends $(n_i -1)$ packets of length $\lmax_i$ followed by a packet of length $\lmodp_i$; this is because at the end of serving these packets the deficits of flow $i$ becomes $\mdelta_i$ and the head-of-the-line packet has a length $\lmax_i > \mdelta_i$. Second, observe that in any other complete services for flow $i$ (if any),  it sends $n_j $ packets of length $\lmax_j$ followed by a packet of length $\lmod_j$. Hence, in the first complete service of flow $i$, which happens after time $s$, flow $i$ is served by $\lp Q_i -  \mdelta_i \rp$; and in every other complete services for flow $i$, it is served by $Q_i$. (the green parts in Fig.~\ref{fig:tight})
	
	Then, by combining the last two paragraphs, observe that (Fig.~\ref{fig:tight})
	\begin{itemize}
		\item Flow $i$ is served in $(s , \tau_p]$ by $\lb pQ_i - \mdelta_i \rb^+$.
		\item Every other flow $j$ has $p+1$ complete services in $(s , \tau_p]$, and they are served by $ \lp (p+1)Q_j + \mdelta_j \rp $.
	\end{itemize}
	It follows that
	\begin{equation}\label{eq:phiTight1}
		D_j(\tau_p) - D_j(s) = \phi_{i,j} \lp D_i(\tau_p) - D_i(s) \rp
	\end{equation}
	Then, there are two cases for $s + \tau $: whether $s + \tau < \tau_p$ or $s + \tau \geq \tau_p$. In the former case, $s + \tau$ is not in the middle of a service for flow of interest and hence $D_i(s + \tau) = D_i(\tau_p)$; in the latter case, $s + \tau$ is in the middle of a service for flow $i$ and $D_i(s + \tau) - D_i(\tau_p) < Q_i$; thus observe that  $ \phi_{i,j} \lp D_i(s + \tau) - D_i(s) \rp =  \phi_{i,j} \lp D_i( \tau_p) - D_i(s) \rp$. Hence, in both cases \eqref{eq:phiTight} holds.
	
	Then, If we apply $\psi^{\downarrow}_i$ to both sides of \eqref{eqn:3}, the right-hand side is equal to $\beta_i^0(\tau)$. Thereby, we should prove
	\begin{equation} \label{eqn:2}
		\psi^{\downarrow}_i\left(D(s + \tau) - D(s) \right) = D_i(s +  \tau) - D_i(s)
	\end{equation}
	Let $y=D(s + \tau) - D(s) $ and $x=D_i(s +  \tau) - D_i(s)$. Our goal is now to prove that
	\begin{equation} \label{eqn:2a}
		\psi^{\downarrow}_i\left(y\right) =x
	\end{equation}
	Again consider the two cases for $s + \tau$.
	
	\textbf{Case 1:} $s + \tau < \tau_{p}$
	
	In this case the scheduler is not serving flow $i$ in $[\tau_p, s + \tau]$; thus $D_i(s + \tau) = D_i(\tau_p)$. 	Combining it with \eqref{eq:phiTight1}, it follows that
	\begin{equation}
		\begin{aligned}
			&\psi_i(x) = x + \underbrace{\sum_{j,j\neq i} \phi_{i,j} \lp x\rp }_{\sum_{j,j\neq i} \lp D_j(\tau_p) - D_j(s) \rp }\\
			& y = x + \sum_{j,j\neq i} \lp D_j(s + \tau) - D_j(s) \rp 
		\end{aligned}
	\end{equation}
	and thus
	\begin{equation}
		\label{eq:jhkds}
		\psi_i(x) \geq y
	\end{equation}
	Let $x-\lmod_i<x'<x$; flow $i$'s output becomes equal to $x'$ during the emission of the last packet thus
	\begin{equation}
		\psi_i(x') = x' + \sum_{j,j\neq i} \lp D_j(\tau_{p-1}) - D_j(s) \rp \\
	\end{equation}
	Hence
	\begin{equation}
		\label{eq:jhkdt}
		\forall x'\in (x-\lmod_i,x),\psi_i(x') < y
	\end{equation}
	Combining \eqref{eq:jhkds} and \eqref{eq:jhkdt} with Lemma~\ref{lsi:lem} shows \eqref{eqn:2a}.
	
	\textbf{Case 2:} $s + \tau \geq \tau_p$
	
	In this case the scheduler is serving flow $i$ in $[\tau_{p} , s + \tau]$.		For every other flow $j$, we have $D_j(s + \tau) = D_j( \tau_p)$. Hence, combining it with \eqref{eq:phiTight},
	\begin{equation}\label{eq:kjlasdf}
		\psi_i(x)=
		D_i(s +  \tau) - D_i(s) +  \sum_{j,j \neq i}\underbrace{ \phi_{i,j}\lp D_i(s +  \tau) - D_i(s)  \rp}_{D_j(s +  \tau) - D_j(s)} =y
	\end{equation}
	As with case 1, for any $x'\in (x - \lmod_i,x)$, we have $\psi_i(x)<y$, which shows \eqref{eqn:2a}.
	
	%
	%
	This shows that \eqref{eqn;tight} holds. It remains to show that the system constraints are satisfied.
	
	\textbf{Step 3: Verifying the Trajectory}
	
	We need to verify that the service offered to the aggregate satisfies the strict service curve constraint.
	Our trajectory has one busy period, starting at time $0$ and ending at some time $T_{\max}\geq \tau$. We need to verify that
	\begin{equation}\label{eq-ljkhsdf}
		\forall t_1, t_2\in [0,T_{\max}] \mwith t_1<t_2, D(t_2)-D(t_1)\geq \beta(t_2-t_1)
	\end{equation}
	
	\textbf{Case 1:} $t_2<s$
	
	Then $D(t_2)-D(t_1)=K(t_2-t_1)$. Observe that, by the Lipschitz continuity condition on $\beta$, for all $t\geq 0$, $\beta(t)=\beta(t)-\beta(0)=\beta(t)\leq Kt$ thus $K(t_2-t_1)\geq \beta(t_2-t_1)$.
	
	\textbf{Case 2:} $t_1<s\leq t_2$
	
	Then $D(t_2)-D(t_1)=\beta(t_2-s)+K(s-t_1)$. By the Lipschitz continuity condition:
	
	\begin{equation}\label{eq-jhkasd}
		\beta(t_2-t_1)-\beta(t_2-s)\leq K(s-t_1)
	\end{equation}
	thus $D(t_2)-D(t_1)\geq\beta(t_2-t_1)$.
	
	\textbf{Case 3:} $s\leq t_1< t_2$
	
	Then $D(t_2)-D(t_1)=\beta(t_2)-\beta(t_1)\geq \beta(t_2-t_1)$ because $\beta$ is super-additive.
	
\end{proof}

\subsection{Proof of Theorem \ref{thm:closeDelay}} \label{sec:proofDelay}
We first prove the following:
\begin{equation} \label{eq:delayConcave}
h \lp \alpha_i, \beta_i^0 \rp = T + \sup_{t \geq 0}\{\frac{1}{c}\psi_i\lp \alpha_i(t)\rp - t\}
\end{equation}
First, as in \cite[Prop. 3.1.1]{le_boudec_network_2001},
\begin{equation}
	h \lp \alpha_i, \beta_i^0 \rp = \sup_{t \geq 0} \{ {\beta_i^0}^\downarrow \lp\alpha_i(t) \rp -t \}
\end{equation}
Second, we show that
\begin{equation}\label{eq:delayWithgamma}
	{\beta_i^0}^\downarrow = T + \frac{1}{c}\gamma_i^\downarrow
\end{equation}
As in \cite[Prop. 7]{boyer_pseudo_inverse}, for two functions $f,g \in \mathscr{F}$ where  $f$ is right-continuous, we have $(f \circ g)^\downarrow = g^\downarrow \circ f^\downarrow$; as $\beta_i^0 = \gamma_i \circ \beta$ and $\gamma_i$ is continuous , it follows that 
\begin{equation}
	{\beta_i^0}^\downarrow  = \beta^\downarrow  \circ \gamma_i^\downarrow 
\end{equation}
Observe that $\beta^\downarrow (x)= \frac{1}{c}x + T$ for $x >0$ and $\beta^\downarrow (x) = 0$ for $x =0$. Combine this with the above equation to conclude \eqref{eq:delayWithgamma}.

Third, observe that $ \gamma_i^\downarrow$ is the left-continuous version of $\psi_i$, and let us denote it by $ \gamma_i^\downarrow = \psi^L_i$. It follows that
\begin{equation}
	h \lp \alpha_i, \beta_i^0 \rp = T + \sup_{t \geq 0}\{\frac{1}{c}\psi^L_i\lp \alpha_i(t)\rp - t\}
\end{equation}
Observe that $\sup_{t \geq 0}\{ \psi^L_i\lp \alpha_i(t)\rp - t \} = \sup_{t \geq 0}\{\psi_i\lp \alpha_i(t)\rp - t\}$ as $\psi_i$ is right-continuous. Therefore, combining it with the above equation, \eqref{eq:delayConcave} is shown.

Let us prove item 1), i.e., assuming $\alpha_i = \gamma_{r_i,b_i}$. Define $H$ as 
\begin{equation}\label{eq:H}
	H(t) \isdef T +\frac{1}{c}\psi_i\lp \alpha_i(t)\rp - t
\end{equation}
Using \eqref{eq:delayConcave}, we have $h \lp \alpha_i, \beta_i^0 \rp = \sup_{t \geq 0} \{ H(t) \}$. By plugging $\alpha_i$ and $\psi_i$ in $H$, we have $H(t)=T + \frac{\sum_{j \neq i}\lp Q_j + \mdelta_j \rp}{c} + \frac{1}{c}\lp r_it + b + \lfloor\frac{r_it + b + \mdelta_i}{Q_i} \rfloor\sum_{j \neq i}Q_j  \rp - t$
Then, as $r_i \leq \frac{Q_i}{Q_\tot}c$, observe that function $H$ is linearly decreasing between $0 \leq t < \tau_i$ with a jump at $t = \tau_i$; also, $H(t) \leq H(\tau_i)$ for all $t \geq \tau_i$ (see the above panel of Fig.\ref{fig:H}). Hence, the supremum of $H$ is obtained either at $t =0$ or $t = \tau_i$. This concludes item 1).
 
 Item 2) can be shown in a similar manner, however, function $H$ is non-decreasing between $0 \leq t < \tau_i$ (see the bottom panel of Fig.\ref{fig:H}).
 
 The same proof holds for item 3).

\begin{figure}[htbp]
	\centering
	\includegraphics[width=\linewidth]{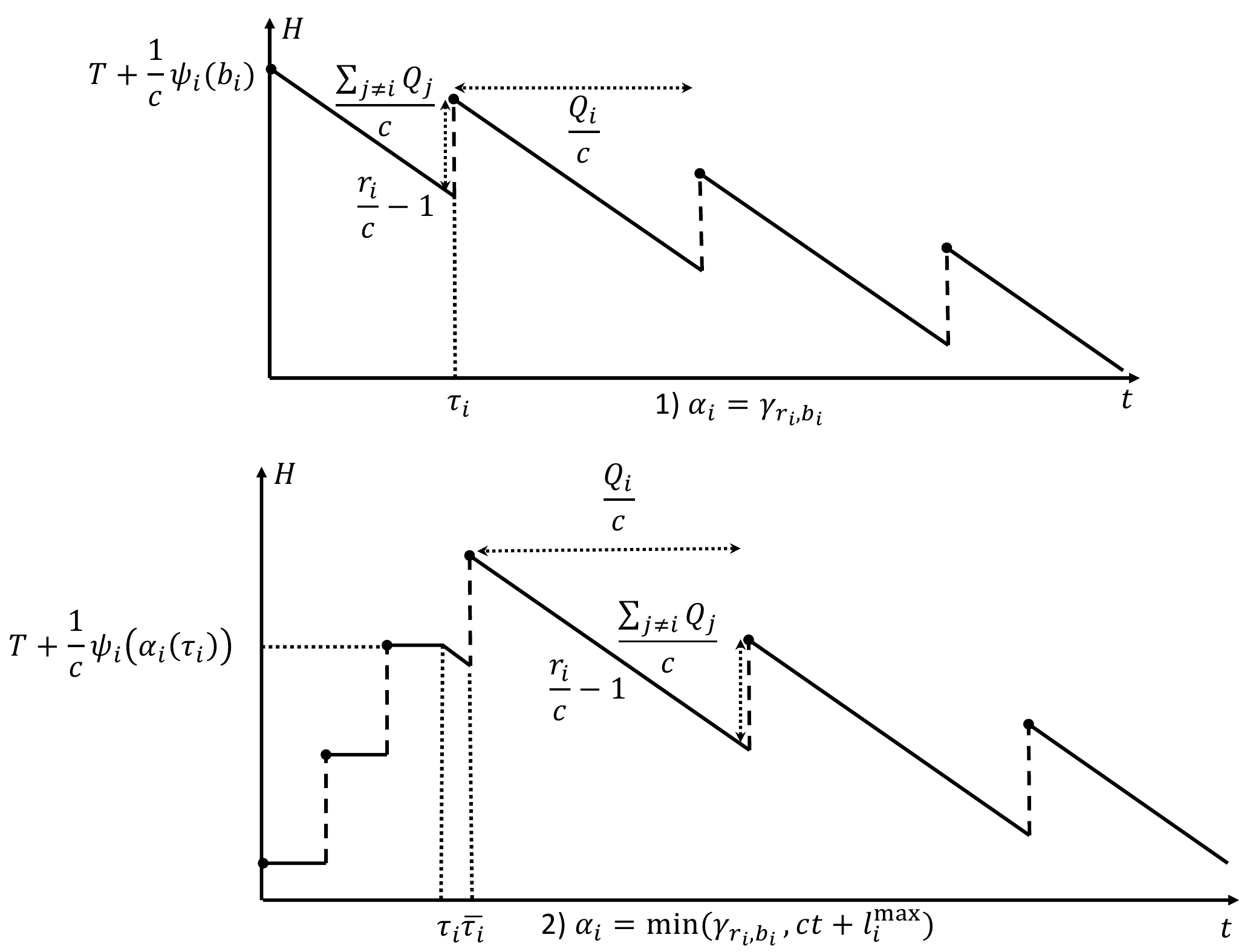}
	\caption{\sffamily \small Illustration of function $H$ defined in \eqref{eq:H}}
	\label{fig:H}
\end{figure}

\subsection{Proof of Theorem \ref{thm:drrServiceApp}} \label{sec:proofApp}

First observe that, since $\phi'_{i,j}\in \mathscr{F}$, it follows that $\beta_i'\in  \mathscr{F}$. Second, as for every $j \neq i$, $\phi_{i,j} \leq \phi_{i,j}^'$, we have $\psi_i \leq \psi_i^'$. In \cite[Sec. 10.1]{liebeherr2017duality}, it is shown that
$
\forall f, g \in \mathscr{F}, f \geq g \Rightarrow f^{\downarrow} \leq g^{\downarrow}
$. 
Applying this with $f = \psi_i'$ and $g = \psi_i$ gives that $\psi_i^{'\downarrow} \leq \psi_i^{\downarrow}$. It follows $\beta_i'(t) = \psi_i^{'\downarrow} \lp \beta(t) \rp \leq \psi_i^{\downarrow} \lp \beta(t) \rp = \beta_i(t)$. The conclusion follows from the fact that any lower bound in $ \mathscr{F}$ of a strict service curve is a strict service curve.
\qed

\subsection{Proof of Theorem \ref{thm:optDrrServiceInd}} \label{sec:proofOptInd}

 First we give a lemma on the operation of DRR, which follows from some of the results in the proof of Theorem~\ref{thm:drrService}. 
\begin{lemma}\label{lem:serviceJ}
	Assume that the output of flows $j \in \bar{J}$ are constrained by arrival curves $\alpha_j^* \in \mathscr{F}$, i.e., $D_j$ is constrained by $\alpha_j^* $. Then, $\gamma^J_i \circ \lb \beta - \sum_{j \in \bar{J}}\alpha_j^* \rb^+_{\uparrow}$ is a strict service curve for flow $i$.
\end{lemma}
\begin{proof}
		Consider a backlogged period $(s,t]$ of flow of interest~$i$.
	As the interval $(s,t]$ is a backlogged period, and since $\beta$ is a strict service curve for the aggregate of flows, we have
	\begin{equation}
		\beta(t-s) \leq (D_i(t) - D_i(s)) + \sum_{j \neq i} \lp D_j(t) - D_j(s) \rp
	\end{equation}
	For $j \in J$, upper bound $\lp D_j(t) - D_j(s) \rp$ as in \eqref{eq:jksa2}, and for $j \in \bar{J}$, upper bound $\lp D_j(t) - D_j(s) \rp \leq \alpha_j^*(t-s)  $, as $\alpha_j^*$ is an arrival curve for $D_j$, to obtain
	\begin{equation}
		\begin{aligned}
				&\beta(t-s) - \sum_{j \in \bar{J}}\alpha_j^*(t-s) \\
				&\leq \underbrace{(D_i(t) - D_i(s))  + \sum_{j \in J}  \phi_{i,j}\lp D_i(t) - D_i(s) \rp }_{\psi_i^J \lp D_i(t) - D_i(s) \rp }
		\end{aligned}
	\end{equation}
	As $\psi_i^J$ is a non-negative function, it follows that 
		\begin{equation}
		\begin{aligned}
			\lb \beta - \sum_{j \in \bar{J}}\alpha_j^* \rb^+(t-s) \leq \psi_i^J \lp D_i(t) - D_i(s) \rp  
		\end{aligned}
	\end{equation}
	As $\psi_i^J$ is an increasing function, it follows that the right-hand side is an increasing function of $(t-s)$.
Then, by applying \cite[Lemma 3.1]{bouillard_deterministic_2018}, it follows that the inequality holds for the non-decreasing closure of the left-hand side (with respect to $t-s$), namely
		\begin{equation}\label{eq:totalserviceJ}
	\begin{aligned}
		 \lb \beta - \sum_{j \in \bar{J}}\alpha_j^* \rb^+_{\uparrow}(t-s) \leq \psi_i^J \lp D_i(t) - D_i(s) \rp  
	\end{aligned}
\end{equation}

	Then, we use the lower pseudo-inverse technique to invert \eqref{eq:totalserviceJ} as in \eqref{lem:lsi}, 
\begin{equation} 
	\begin{aligned}
		D_i(t) - D_i(s) \geq \psi_i^{J\downarrow} \lp   \lb \beta - \sum_{j \in \bar{J}}\alpha_j^*  \rb^+_{\uparrow} (t-s) \rp
	\end{aligned}
\end{equation}
Hence, the right-hand side is a strict service curve for flow $i$. Observe that $\gamma_i^J = \psi_i^{J\downarrow}$. 
\end{proof}

	As $\oldservice_j$ is a strict service curve for flow $j$, it follows that $\alpha_j \oslash \oldservice_j$ is an arrival curve for the output of flow $j$. Then, for every $J \subseteq N_i$, apply Lemma \ref{lem:serviceJ} with $\alpha_j^* = \alpha_j \oslash \oldservice_j$ for $j \in \bar{J}$ and conclude that $\gamma^J_i \circ \lb \beta - \sum_{j \in \bar{J}}\lp \alpha_j \oslash \oldservice_j \rp \rb^+_{\uparrow}$ is a strict service curve for flow $i$. Lastly, the maximum over all $J$ is also a strict service curve for flow $i$.

\qed
 \subsection{Proof of Corollary \ref{thm:optDrrService}}\label{sec:proofoptDRR}
 We proceed the proof by showing that for every flow $i$, $\newserviceBar_i \leq \newservice_i$. 	 Fix $i \leq n$ and $t \geq 0$. We want to show that $\newserviceBar_i(t) \leq \newservice_i(t)$. Let $J^*$ be $\{j \in N_i~|~ \phi_{i,j}\lp\oldservice_i(t)\rp  < \lp \alpha_j \oslash \oldservice_j \rp (t)\}$. Then, observe that
 \begin{equation}
 	\newservice_i(t) \geq \gamma_i^{J^*} \lp \beta(t) - \sum_{j \in \bar{J}^*}\lp \alpha_j \oslash \beta_j \rp (t) \rp  
 \end{equation}
 Apply $\psi_i^{J^*}$ to the both side and observe that
 \begin{equation} \label{eq:eq1}
 	\psi_i^{J^*} \lp \newservice_i(t) \rp \geq  \beta(t) - \sum_{j \in \bar{J}^*}\lp \alpha_j \oslash \newservice_i \rp (t) 
 \end{equation}
 as $\newservice_i \geq \oldservice_i(t)$ and $\phi_{i,j}$ is increasing, it follows that 
 \begin{equation} \label{eq:eq2}
 	\sum_{j \in \bar{J}^*} \phi_{i,j} \lp \newservice_i(t)\rp  \geq 	\sum_{j \in \bar{J}^*} \phi_{i,j} \lp \oldservice_i(t)\rp  
 \end{equation}
 Then, sum the both side of \eqref{eq:eq1} and \eqref{eq:eq2} to obtain
 \begin{equation} 
 	\psi_i \lp \newservice_i(t) \rp \geq  \beta(t) + \underbrace{\sum_{j \in \bar{J}^*} \lp \phi_{i,j} \lp \beta_i(t)\rp - \lp \alpha_j \oslash \oldservice_j \rp (t) \rp}_{\sum_{j \neq i} \lb \phi_{i,j} \lp \beta_i(t)\rp - \lp \alpha_j \oslash \beta_j \rp (t) \rb^+}
 \end{equation}
 Then, by applying \cite[Lemma 3.1]{bouillard_deterministic_2018}, it follows that the above inequality holds for the non-decreasing closure of the left-hand side. Thus,
 \begin{equation}
 	\psi_i \lp \newservice_i(t) \rp \geq  \lp \beta +\sum_{j \neq i} \lb \phi_{i,j} \circ \beta_i - \lp \alpha_j \oslash \oldservice_j \rp \rb^+ \rp_{\uparrow}(t)
 \end{equation}
 Lastly, we use the lower pseudo-inverse technique to invert as in \eqref{lem:lsi} and as $\psi_i^{\downarrow} = \gamma_i$
 \begin{equation}
 \newservice_i(t)  \geq  \underbrace{\gamma_i \circ \lp \beta +\sum_{j \neq i} \lb \phi_{i,j} \circ \beta_i - \lp \alpha_j \oslash \oldservice_j \rp \rb^+ \rp_{\uparrow}(t)}_{\newserviceBar_i(t)}
 \end{equation}
which concludes the proof. \qed

\end{document}